\documentclass{article} 

\usepackage{preprint_arxiv}
\usepackage[utf8]{inputenc} 
\usepackage[T1]{fontenc}    
\usepackage{amsfonts}       
\usepackage{amsmath}
\usepackage{amssymb}
\usepackage{mathtools}
\usepackage{amsthm}
\usepackage{bm}
\usepackage{color}
\usepackage{soul}           
\usepackage{url}
\usepackage{booktabs}       
\usepackage{nicefrac}       
\usepackage{enumerate}
\usepackage{graphicx}
\DeclareGraphicsExtensions{.pdf,.png,.jpg}
\graphicspath{{figures/}}
\usepackage{comment}
\usepackage{hyperref}
\usepackage{centernot}     
\usepackage{algorithm}
\usepackage[noend]{algpseudocode}

\usepackage{float}
\newcommand{\R}{\mathbb{R}}

\newcommand{\BigO}[1]{\ensuremath{\mathcal{O}\left(#1\right)}}                  
\newcommand{\BigOm}[1]{\ensuremath{\Omega\left(#1\right)}}                      
\newcommand{\poly}{\mathrm{poly}}
\newcommand{\vect}[1]{\ensuremath{\mathbf{#1}}}                                 
\newcommand{\KL}[2]{\ensuremath{\mathbb{KL} \left(#1 \middle\Vert #2 \right)}}   
\newcommand{\Exp}[2][]{\ensuremath{\mathbb{E}_{#1}\left[#2\right]}}                
\newcommand{\Var}[2][]{\ensuremath{\mathrm{Var}_{#1}\left[#2\right]}}                
\newcommand{\Norm}[1]{\ensuremath{\left\lVert #1 \right\rVert}}                  
\newcommand{\NormI}[1]{\ensuremath{\left\lVert #1 \right\rVert}_1}               
\newcommand{\NormII}[1]{\ensuremath{\left\lVert #1 \right\rVert}_2}              
\newcommand{\NormInfty}[1]{\ensuremath{\left\lVert #1 \right\rVert_{\infty}}}    
\newcommand{\NormMax}[1]{\ensuremath{\left\lvert #1 \right\rvert_{\mathrm{max}}}}    
\newcommand{\parg}{\makebox[1ex]{$\mathbf{\cdot}$}}                              
\newcommand{\Tp}[1]{{#1}^{\top}}												    
\newcommand{\Tpn}[1]{{#1}^{-\top}}	                                             
\newcommand{\mquad}{\kern-1em}													
\newcommand{\Ineq}[2][]{\overset{#1}{#2}}                                         

\newcommand{\matrx}[1]{\begin{bmatrix}#1\end{bmatrix}}                           

\newcommand{\InNorm}[1]{{\left\vert\kern-0.2ex\left\vert\kern-0.2ex\left\vert #1 
    \right\vert\kern-0.2ex\right\vert\kern-0.2ex\right\vert}}                    

\newcommand{\InNormII}[1]{{\left\vert\kern-0.2ex\left\vert\kern-0.2ex\left\vert #1 
    \right\vert\kern-0.2ex\right\vert\kern-0.2ex\right\vert}_2}                    

\newcommand{\InNormInfty}[1]{{\left\vert\kern-0.2ex\left\vert\kern-0.2ex\left\vert #1 
    \right\vert\kern-0.2ex\right\vert\kern-0.2ex\right\vert}_{\infty}}           

\newcommand{\Abs}[1]{\ensuremath{\left \lvert #1 \right \rvert}}                 
\newcommand{\Prob}[2][]{\ensuremath{\mathrm{Pr}_{#1}\left\{ #2 \right\}}}        
\DeclarePairedDelimiterX{\Inner}[2]{\langle}{\rangle}{#1, #2}                    

\newcommand{\MI}{\mathnormal{I}}                                                     

\newcommand{\what}[1]{\widehat{#1}}                                              
\newcommand{\defeq}{\overset{\mathrm{def}}{=}}                                                      

\newcommand{\Set}[1]{\{#1\}}                                                     
\newcommand{\emset}{\varnothing}                                                 
\DeclareMathOperator*{\union}{\cup}
\DeclareMathOperator*{\intersection}{\cap}

\newcommand{\floor}[1]{\left\lfloor#1\right\rfloor}
\newcommand{\ceil}[1]{\left\lceil#1\right\rceil}

\definecolor{gray}{rgb}{0.7,0.7,0.7}

\DeclareMathOperator*{\argmin}{argmin}

\newtheorem{proposition}{Proposition}
\newtheorem{assumption}{Assumption}
\newtheorem{problem}{Problem}
\newtheorem{lemma}{Lemma}
\newtheorem{theorem}{Theorem}
\newtheorem{remark}{Remark}
\newtheorem{corollary}{Corollary}
\theoremstyle{definition}

\def\independenT#1#2{\mathrel{\rlap{$#1#2$}\mkern2mu{#1#2}}}
\DeclareMathOperator{\independent}{\protect\mathpalette{\protect\independenT}{\perp}}  
\DeclareMathOperator{\notindependent}{\centernot{\independent}}                        

\newcommand{\Diag}{\vect{Diag}}
\newcommand{\trace}{\mathrm{tr}}

\newcommand{\Erdos}{Erdős} 
\newcommand{\Renyi}{Rényi }
\newcommand{\B}[1]{B^{\scriptscriptstyle (#1)}}
\newcommand{\Di}[1]{D^{\scriptscriptstyle (#1)}}
\newcommand{\Bt}[1]{\widetilde{B}^{\scriptscriptstyle (#1)}}
\newcommand{\Dt}[1]{\widetilde{D}^{\scriptscriptstyle (#1)}}
\newcommand{\wtl}[1]{\widetilde{#1}}

\newcommand{\G}[1]{G^{\scriptscriptstyle (#1)}}
\newcommand{\Gt}[1]{\widetilde{G}^{\scriptscriptstyle (#1)}}

\newcommand{\var}{\sigma^2}
\newcommand{\vart}{\widetilde{\sigma}^2}
\newcommand{\mi}{-\mathrm{i}}

\newcommand{\D}{\Delta}
\newcommand{\DB}{{\Delta_{B}}}

\newcommand{\DO}{{\Delta_{\Omega}}}
\newcommand{\DOh}{{\widehat{\Delta}_{\Omega}}}
\newcommand{\DOt}{{\Delta_{\Omega}^*}}

\newcommand{\X}[1]{X^{{\scriptscriptstyle (#1)}}}

\newcommand{\Om}[1]{\Omega^{{\scriptscriptstyle (#1)}}}

\newcommand{\Sm}[1]{\Sigma^{{\scriptscriptstyle (#1)}}}
\newcommand{\Smh}[1]{\widehat{\Sigma}^{{\scriptscriptstyle #1}}}

\newcommand{\smt}[1]{\widetilde{\sigma}^{{\scriptscriptstyle (#1)}}}
\newcommand{\Par}[1]{\pi^{\scriptscriptstyle (#1)}}
\newcommand{\nPar}{\pi}   
\newcommand{\Chi}[1]{\phi^{\scriptscriptstyle (#1)}}

\newcommand{\corr}[1]{\mathrm{corr}^{\scriptscriptstyle(#1)}} 
\newcommand{\supp}{\mathrm{supp}}
\newcommand{\Anc}{\mathcal{A}}                                
\newcommand{\Anct}{\widetilde{\mathcal{A}}}                                
\newcommand{\lAnc}[1]{{\mathcal{A}}^{\scriptscriptstyle(#1)}} 
\newcommand{\computeOrder}{\textproc{ComputeOrder}}
\newcommand{\orientEdges}{\textproc{orientEdges}}
\newcommand{\prune}{\textproc{prune}}
\newcommand{\vectorize}{\mathbf{\mathrm{vec}}}
\newcommand{\perr}{p_{\mathrm{err}}}
\newcommand{\Gf}{\mathcal{G}}  
\newcommand{\Gfr}{\widetilde{\mathcal{G}}}  
\newcommand{\DGf}{\mathcal{G}_{\Delta}}  
\newcommand{\DGfr}{\widetilde{\mathcal{G}}_{\Delta}}  
\newcommand{\Dec}{\zeta}
\newcommand{\Pf}{\mathcal{P}}  
\newcommand{\Xf}{\mathfrak{X}} 
\newcommand{\bh}{\widehat{b}}
\newcommand{\betah}{\widehat{\beta}}
\newcommand{\kdmin}{K^{\mathrm{d}}_{\mathrm{min}}}
\newcommand{\komax}{K^{\mathrm{o}}_{\mathrm{max}}}
\newcommand{\kdminS}{K^{\mathrm{d,S}}_{\mathrm{min}}}
\newcommand{\komaxS}{K^{\mathrm{o,S}}_{\mathrm{max}}}
\newcommand{\Sc}{S^c}
\newcommand{\xb}{\bar{x}}
\newcommand{\eigmax}{\lambda_{\mathrm{max}}}
\newcommand{\eigmin}{\lambda_{\mathrm{min}}}
\newcommand{\skel}{\mathrm{skel}}
\newcommand{\dt}{d'}

\title{Direct Learning with Guarantees of the Difference DAG Between Structural Equation Models}

\author{%
Asish Ghoshal\thanks{Correspondence to: \texttt{aghoshal@purdue.edu}}
\qquad
Kevin Bello 
\qquad Jean Honorio\\
Department of Computer Science\\
Purdue University\\
West Lafayette, IN 47906, USA
}

\begin{document}    
\maketitle

\begin{abstract}
Discovering cause-effect relationships between variables from observational data is a fundamental challenge in many scientific disciplines.
However, in many situations it is desirable to directly estimate the change in causal relationships across two different conditions, e.g., estimating the change in genetic expression across healthy and diseased subjects can help isolate genetic factors behind the disease.
This paper focuses on the problem of  \textit{directly} estimating the structural difference between two \textit{structural equation models} (SEMs), having the same topological ordering,  given two sets of samples drawn from the individual SEMs. 
We present an \textit{principled} algorithm that can recover the \textit{difference SEM} in $\BigO{d^2 \log p}$ samples, where $d$ is related to the number of edges in the difference SEM of $p$ nodes. 
We also study the \textit{fundamental limits} and show that \textit{any} method requires at least $\Omega(\dt \log (\nicefrac{p}{\dt}))$ samples to learn difference SEMs with at most $\dt$ parents per node.
Finally, we validate our theoretical results with synthetic experiments and show that our method outperforms the state-of-the-art.
Moreover, we show the usefulness of our method by using data from the medical domain.
\end{abstract}

\section{Introduction and Related Work}
Discovering causal relationships from observational studies is of tremendous importance in many scientific disciplines. 
In Pearl's framework of causality \cite{pearl2009causality}, such cause-effect relationships are modeled using directed acyclic graphs (DAGs).
One of the central problems in causal inference is then to recover a DAG of cause-effect relationships over variables of interest, given observations of the variables.
It is well known that the number of samples needed to recover a DAG over $p$ variables and maximum number of neighbors $d$ grows as $\Theta(\poly(d) \log p)$ \cite{ghoshal2017information, ghoshal2018learning}.
Therefore, the presence of hub nodes makes it especially challenging to recover the DAG from a few samples.
In many situations however, the \textit{changes in causal structures} across two different settings is of primary interest. 
For instance, the changes in structure of the gene regulatory network between cancerous and healthy individuals might help shed light on the genetic causes behind the particular cancer. 
In this case, estimating the individual networks over healthy and cancerous subjects is not sample-optimal since many background genes do not change across the subjects
or even across distant species \cite{tanay2005conservation}.
While the individual networks might be \textit{dense}, the difference between them might be \textit{sparse}.

In this paper, we focus on the problem of learning the structural differences between two linear structural equation models (SEMs) (or Bayesian networks) given samples drawn from each of the model. 
We assume that the (unknown) topological ordering between the two SEMs remains consistent, i.e., there are no edge reversals. 
This is a reasonable assumption in many settings and have also been considered by prior work \cite{wang_direct_2018}. 
For instance, edges representing genetic interactions may appear or disappear or change weights, but generally do not change directions \citep{belyaeva2020dci}.
Furthermore, in the multi-task learning literature it is often assumed that
the noise variances across different tasks are the same (c.f. \cite{jalali_dirty_2010}). 
Our primary goal in this paper is to develop an algorithm that \textit{directly} learns the difference DAG by using a number of samples that depends only on the \textit{sparsity} of the difference DAG. 
This is a much more challenging problem than structure learning of Bayesian networks since in the latter case when the causal ordering is known, structure learning boils down to regressing each variable against all other variables that come before it in the topological order and picking out the non-zero coefficients.
However, the fact that the individual DAGs are \textit{dense}, rules out performing regressions in the individual model and then comparing invariances of the coefficients across the two models.

The problem of learning the difference between undirected graphs (or Markov random fields) has received much more attention than the directed case. 
For instance \cite{zhao_direct_2014, liu2017support, yuan_differential_2017, fazayeli_generalized_2016} develop algorithms for estimating the difference between Markov random fields and Ising models with finite sample guarantees.
Another closely related problem is estimating invariances between causal structure across multiple environments \cite{peters2016causal}. 
However, this is desirable when the \emph{common structure} is expected to be sparse across environments, as opposed to our setting where the \textit{difference} is expected to be sparse.

The problem of estimating the difference between DAGs has been previously considered by \cite{wang_direct_2018}, who developed a PC-style algorithm \cite{spirtes2000causation}, which they call \emph{DCI}, for learning the difference between the two DAGs by testing for invariances between regression coefficients and noise variances between the two models. 
However, sample complexity guarantees are hard to obtain for their method due to the use of many approximate asymptotic distributions of test statistics. 
Since the primary motivation behind directly estimating the difference between two DAGs is sample-efficiency, a lack of finite sample guarantees is a significant shortcoming. 
In contrast, our algorithm works by repeatedly eliminating vertices and re-estimating the difference of precision matrix over the remaining vertices. 
Thereby, we are able to leverage existing algorithms for computing the difference of precision matrix to obtain finite sample guarantees for our method.  Furthermore, the DCI algorithm estimates regression coefficients (and noise variances) in the individual DAGs,  while our method never estimates weights or noise variances of individual SEMs. Consider the example given in Figure \ref{fig:example} where the difference DAG contains
only one edge $X_4 \rightarrow X_2$. In order to prune the edges $X_1 - X_4$ and $X_3 - X_4$ which are present in the difference undirected graph but not in the difference DAG, DCI would compute regression coefficients  $\theta^{1}_{4 | S}$ and $\theta^{2}_{4 | S}$ for all $S \subseteq \Set{1, 2, 3}$, where $\theta^{1}_{j | S}$ (resp. $\theta^{2}_{j | S}$) denotes regression coefficients obtained by regressing $X_j$ against $X_S$ in the first (resp. second) SEM. 
For linear SEMs, estimating regression coefficients is equivalent to estimating the precision matrix (Lemma 1 in \cite{ghoshal2017learning}).
Furthermore, \citet{danaher2014joint} have shown that directly estimating the difference between precision matrices is more sample efficient  than estimating individual precision matrices and computing the difference. 
\begin{figure}
    \centering
    \includegraphics[width=0.7\linewidth]{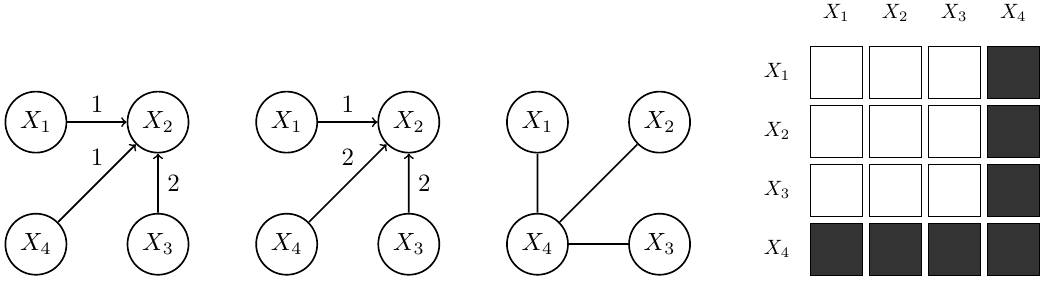}
    \caption{From left to right: the two SEMs, the difference undirected graph (or difference of moral graphs), and the difference of precision matrices between the two SEMs with non-zero entries shown in black. \label{fig:example}}
\end{figure}

\textbf{Our contributions.}
The above results therefore call
for methods that estimate the difference DAG directly without computing individual SEM parameters. Towards that end, we make the following four contributions in this paper: 
\begin{enumerate}
\item We are the first to obtain high-dimensional finite sample guarantees for the problem of directly estimating the structural changes between two linear SEMs. When the noise variances of the variables across the DAGs are the same
our algorithm recovers the difference DAG in $\BigO{d^2 \log p}$ samples where  $d$ is the maximum number of edges in the difference of the moralized sub-graphs of the two SEMs, here the maximum is computed over subsets of variables. 
In the general (unequal noise variance) case, our method returns a partially directed DAG with the correct skeleton and orientation of the directed edges.

\item We show that, under some incoherence conditions, our method is strictly more sample efficient than estimating the individual SEMs and the DCI algorithm \cite{wang_direct_2018}.

\item Our algorithm improves upon the computational complexity of the algorithm by \cite{wang_direct_2018} for direct estimation of DAGs in the sense that it tests for the presence of fewer edges in the difference DAG.

\item We show that any method requires at least $\BigOm{d' \log(\nicefrac{p}{d'})}$ samples to recover the difference DAG, where $d'$ is the maximum number of parents of a variable in the difference DAG thereby showing that our method is sample optimal in the number of variables.
\end{enumerate}

\section{Notation and Problem Statement}
Let $X = (X_1, \ldots, X_p)$ be a $p$-dimensional vector. 
Let $[p]$ denote the set of integers $\{1\ldots p\}$.
We will denote a structural equation model by the tuple $(B, D)$ where $B$
is an autoregression matrix and $D = \Diag(\Set{\var_i})$ is a diagonal matrix of noise variances.
The SEM $(B, D)$ defines the following generative model over $X$:
\begin{align*}
X_i = B_{i*} X  + \varepsilon_i, \quad (\forall i \in [p]) ,
\end{align*}
$\B{1}_{i,i}=0$, $\Exp{\varepsilon} = 0$ and $\Var{\varepsilon} = \var_i < \infty$. We will denote the
$i$-th row (resp. $i$-th column) of a matrix $A$ by $A_{i*}$ (resp. $A_{*i}$).
Note that our notation of a linear SEM disregards the distribution the noise variables and only considers their second moment as our algorithm only utilizes the second moment of variables.
An autoregression matrix $B$ encodes a DAG $G = ([p], \supp(B))$ over $[p]$,
where $\supp(\parg)$ denotes the support set of a matrix (or a vector), i.e., 
$\supp(B) = \Set{(i,j) \in [p] \times [p] \mid B_{i,j} \neq 0}$.
Note that the edge $(i,j)$ denotes the directed edge $i \leftarrow j$.

Given two SEMs, $(\B{1}, \Di{1})$ and $(\B{2}, \Di{2})$, our goal in this paper is to recover the structural difference between the two DAGs, i.e., $\supp(\DB) \defeq \supp(\B{1} - \B{2})$. We assume that each of the individual autoregression matrices ($\B{1}$ and $\B{2}$)  to be potentially dense but their difference to be sparse. 
Specifically, we assume that each row and column of $\D_B \defeq \B{1} - \B{2}$ to have at most $d$ $(\ll p)$ non-zero entries. We further assume that there are no edge reversals between $\B{1}$ and $\B{2}$, thereby resulting in $\supp(\DB)$ being a DAG, which, as stated in the introduction, is a reasonable assumption in several practical problems \citep{zhao_direct_2014,belyaeva2020dci}. 
Formally, we are interested in the following problem:
\begin{problem}
Given two sets of observations $\X{1} \in \R^{n_1 \times p}$ and $\X{2} \in \R^{n_2 \times p}$,
drawn from the \emph{unknown} SEMs $(\B{1}, \Di{1})$ and $(\B{2}, \Di{2})$ respectively, estimate $\supp(\DB)$.
\end{problem}
We will often index the two SEMs by $\kappa \in \Set{1,2}$. 
We will denote the set of parents of the $i$-th node in the SEM indexed by $\kappa$ by $\Par{\kappa}(i)$,
while the set of children are denoted by $\Chi{\kappa}(i)$. 
We will denote the difference between the precision matrices of the two SEMs by: $\D_{\Omega}$, and the precision matrix
over any subset of variables $S \subseteq [p]$ by $\D_{\Omega}^S$. Similarly, $\Om{\kappa, S}$ denotes the precision
matrix over the subset $S$ in the SEM indexed by $\kappa$. We will denote the set of topological ordering
induced by a DAG $G = ([p], E)$ by $\mathcal{T}(G) = \Set{(\tau_1, \ldots, \tau_p) \in \Pi ([p]) \mid (\tau_i, \tau_j) \in E \text{ if } i < j}$, where $\Pi([p])$ is the set of permutations of $[p]$.
The notation $i \preceq_{\tau} j$ denotes that the vertex $i$ comes before $j$ (or $i = j$) in the topological order $\tau$.
Finally, we will always index precision matrices by vertex labels, i.e., $\Omega_{i,j}$ denotes the precision
matrix entry corresponding to the $i$-th and $j$-th node of the graph.
\section{Results}
We present a series of results leading up to our main algorithm for direct estimation of the difference between two DAGs. The following result characterizes
the terminal (or sink) vertices, i.e., vertices with no children, of the difference DAG
in terms of the entries of the difference of precision matrix.
\begin{proposition}
\label{prop:terminal_vertex}
Given two SEMs $(\B{1}, \Di{1})$ and $(\B{2}, \Di{2})$ and with 
precision matrices $\Om{1}$ and $\Om{2}$ respectively.
If for any node $i$ the edges incident on its children, and the corresponding noise variances remain
invariant across the two SEMs, i.e., $\forall j \in [p] : \: \B{1}_{j,i} = \B{2}_{j,i}$, $\Di{1}_{i,i} = \Di{2}_{i,i}$
 and $\Di{1}_{j,j} = \Di{2}_{j,j}$, 
then $(\D_{\Omega})_{i,i} = 0$. Furthermore, $i$ is a terminal vertex in the difference
DAG $G = ([p], \D)$ with $\D = \supp(\B{1} - \B{2})$. 
\end{proposition}
The next result characterizes the edge weights and noise variances of the SEM obtained by removing a vertex $i$,
and plays a crucial role in developing our algorithm.
\begin{lemma}
\label{lemma:delete}
Let $(B, D)$ be a SEM with $D = \Diag(\Set{\var_i})$ and DAG $G$. Then the SEM obtained by removing a subset 
of vertices $U \subset [p]$, i.e., the SEM over $X_{[p] \setminus U}$, is given by $(\wtl{B}, \wtl{D})$
with $\wtl{D} = \Diag(\Set{\vart_j}_{j \in [p] \setminus U})$ and
\begin{align*}
\vart_j = \sigma^4_j \left\{\var_j - B_{j,U_j}(\Omega_{U_j, U_j}^{\Anc_j})^{-1}\Tp{(B_{j,U_j})} \right\}^{-1}, \quad
\wtl{B}_{j,k} = \frac{\vart_j}{\var_j}  \left \{ B_{j,k} -  B_{j,U_j}( \Omega_{U_j, U_j}^{\Anc_j})^{-1} (\Omega_{U_j, k}^{\Anc_j} ) \right\}
\end{align*}
$\forall j \in [p] \setminus U$ and $k \in \Anc_j$, where $\Anc_j$ denotes the ancestors of node $j$.
Also, $U_j = \Anc_j \intersection U$, and $\Omega^{\Anc_j}$ is  the precision matrix over $X_{\Anc_j}$. 
Finally, for $k \notin \Anc_j$, $\wtl{B}_{j,k} = 0$. 
\end{lemma}
As a corollary of the above lemma, we have the following result when $U = \Set{i}$ and
$i$ is a terminal vertex, i.e., $\Chi{G}(i) = \emset$.
\begin{corollary}
Given a SEM $(B, D)$ with $D = \Diag(\Set{\var_i})$, if $i$ is a terminal vertex, then the SEM obtained by removing vertex $i$ is given by $(B_{\mi, \mi}, D_{\mi, \mi})$.
\end{corollary}
With the above results in place, we are now ready to state our algorithm for learning the difference DAG. 
At a high level, the algorithm works as follows. Given the difference of precision matrix, we first remove the invariant vertices, i.e., vertices for which the corresponding rows and columns in the difference of precision matrix is all zeros. 
These vertices have no neighbors in the difference DAG and their noise variances remains the same across the two DAGs.
Next, we estimate the topological ordering over the remaining vertices in the difference DAG. 
When the noise variances are the same, the algorithm returns a topological order over all the variables. 
Whereas when some of the noise variances differ between the two SEMs, then the topological order is computed over a subset of variables.
After estimating the topological order, we orient the edges present in the difference of precision matrix according to the  ordering to compute a super-graph of the difference DAG. 
We then perform a final pruning step to remove the ``extra'' edges to obtain the correct difference DAG. We show how all these steps can be performed by manipulating \textbf{only} the difference of precision matrix. 
Furthermore, estimating the topological order affords us with significant computational and statistical advantage as we will elaborate later.
\begin{figure*}[!tb]
\begin{minipage}{0.5\textwidth}
\begin{algorithm}[H]
\caption{Main algorithm\label{alg:exact_alg}}
\begin{algorithmic}[1]
\Require $\Sm{1}$ and $\Sm{2}$
\Ensure $\D$
\State Estimate $\DO$
\State $V \gets [p]$
\State $U \gets \Set{i \mid (\DO)_{i,*} = 0}$ \label{line:U} \Comment{Set of invariant vertices}
\State $\DO \gets (\DO)_{U^c, U^c}$
\State $\Sm{1} \gets \Sm{1}_{U^c, U^c}$;\quad $\Sm{2} \gets \Sm{2}_{U^c, U^c}$
\State   $V \gets V_{U^c}$
\State $O, R \gets \textsc{computeOrder}(\DO, \Sm{1}, \Sm{2}, V)$
\State $\D \gets \Call{orientEdges}{\DO, O, R}$
\State $\D \gets \Call{prune}{\D, \DO, \Sm{1}, \Sm{2}, O, R}$
\State \Return $\D$
\end{algorithmic}
\end{algorithm}
\end{minipage}\hfill
\begin{minipage}{0.48\textwidth}
	\begin{algorithmic}[1]
		\Function{computeOrder}{$\DO$, $\Sm{1}$, $\Sm{2}$, $V$}
		\State $O \gets \emset$ \Comment{Topological order} 
		\While{$\Abs{V} > 1$}
    		\State $S \gets \Set{i \mid (\DO)_{i,i} = 0}$
    		\If{$S = \emptyset$}
    		    \State \Return $O$, $V$
    		\EndIf
    		\State Add $S$ to end of $O$
    		\State Remove the vertices in $S$ from $V$
    		\State Re-estimate $\DO$ over $V$
		\EndWhile 
		\State Add $V$ to end of $O$
		\State \Return $O$, $\emset$
		\EndFunction
	\end{algorithmic}
\end{minipage}

\begin{minipage}{0.48\textwidth}
\begin{algorithmic}[1]
\Function{orientEdges}{$\DO$, $O$, $R$}
\State $\D \gets \emset$
\For{$S \in O$}
    \For{$i \in S$}
        \State $N_i \gets \Set{j \neq i \mid (\DO)_{i,j} \neq 0}$
        \For{$j \in N_i$} 
            \State \textbf{if} $(j,i) \notin \D$ and $j \notin S$ \textbf{then} add $(i,j)$ in $\D$
        \EndFor
    \EndFor
\EndFor
\State $\D \gets \D \union \supp((\DO)_{R, R})$
\State \Return $\D$
\EndFunction
\end{algorithmic}

\end{minipage}\hfill
\begin{minipage}{0.48\textwidth}
\begin{algorithmic}[1]
\Function{prune}{$\D$, $\DO$, $\Sm{1}$, $\Sm{2}$, $O$, $R$}
\For{$(i,j) \in \D$}
\State Let $S_{ij}$ be descendants of $j$ in $O$ except $i$
\State \textbf{if} $(j,i) \in \D$ \textbf{then} $S_{ij} \gets S_{ij} \union R$
\State \textbf{for} $S \subset S_{ij}$ \textbf{do} estimate $\D_{\Omega}^S$
\State \textbf{if} $(\D_{\Omega}^S)_{i,j} = 0$ \textbf{then} remove $(i,j)$ (and $(j,i)$ if present) from $\D$
\EndFor 
\State \Return $\D$
\EndFunction
\end{algorithmic}
\end{minipage}
\end{figure*}

First, we prove the correctness of Algorithm \ref{alg:exact_alg} in the population setting, i.e., when
$\Sm{\kappa}$ is the true covariance matrix of the SEM $(\B{\kappa}, \Di{\kappa})$ for $\kappa \in \Set{1,2}$.
In this case $\DO$ can be computed efficiently by solving the linear system: 
$\Sm{1}(\DO)\Sm{2} = \Sm{2} - \Sm{1}$ \citep{zhao_direct_2014}. Since $\Sm{\kappa}$ is positive definite,
the above system has a unique solution. To prove the correctness of our algorithm in the population setting
we need the following assumption.

\begin{assumption}
\label{ass:main}
Let $(\B{1}, \Di{1})$ and $(\B{2}, \Di{2})$ be two SEMs with the difference DAG given by $\D_G = ([p], \D)$, where $\D = \supp(\B{1} - \B{2})$, and difference of precision matrix given by $\DO$.
Let $U = \Set{i \in [p] \mid (\DO)_{i,*} = 0}$ and let $V = [p] \setminus U$.
Then the two SEMs satisfy the following
assumptions:
\begin{enumerate}[(i)]
\item For $i \in U$, the edges and noise variances are invariant.
\item Let $\mathcal{S} = \Set{ \Set{\tau_1, \ldots, \tau_k} \subseteq [p] \mid \tau \in \mathcal{T}(\D_G), 0 \leq k \leq \Abs{V}}$.
For each $(i,j) \in \D$, and $\forall S \in \mathcal{S} : i,j \in S$, we have that $\corr{1}(X_i, X_j | X_{S'}) \neq \corr{2}(X_i, X_j | X_{S'})$, where $S' = S \setminus \Set{i,j}$.
\end{enumerate}
\end{assumption}
In the above, $\corr{1}(\parg)$ (resp. $\corr{2}(\parg)$) denotes partial correlation in the first (resp. second) SEM.
Condition (i) in the above assumption essentially requires that if none of the (undirected) edges incident on a vertex change in
the moral graph of the two DAGs then the (directed) edges incident on the node remains invariant across the two DAGs.
Condition (ii) above is essentially a restricted version of the faithfulness assumption which requires that if an edge $(i,j)$ changes across the two DAGs then $X_i \notindependent X_j \mid X_{S'}$.

The following theorem certifies the correctness of our algorithm in the \emph{population} case,
where $\skel(S)$ of a set of edges $S$ denotes the undirected skeleton where the directed
edges are converted to undirected edges, i.e., for all $(i,j) \in S$, $ (j,i) \in \skel(S) $.
\begin{theorem}
\label{thm:population}
Let $\D_G = ([p], \D^*)$ be the true difference DAG, where $\D^* = \supp(\B{1} - \B{2})$.
Given the true covariance matrices $\Sm{1}$ and $\Sm{2}$, Algorithm \ref{alg:exact_alg} returns $\D$ such that 
$\skel(\D) = \skel(\D^*)$ and all the directed edges in $\D$ are correctly oriented.
Moreover, if $\Di{1} = \Di{2}$ then $\D = \D^*$.
\end{theorem}
\begin{proof}[Proof Sketch of Theorem \ref{thm:population}]
Let $V' = [p] \setminus U$, where $U$ is the set of invariant vertices defined
in line \ref{line:U} of the main algorithm.
Let $\mathcal{T}(\D_G)$ be the set of topological orderings
induced by the true difference DAG. The correctness of Algorithm \ref{alg:exact_alg} 
follows from the following claims (proved in detail in Appendix \ref{app:proofs}):

\textbf{Claim (i):}  Denote the SEMs obtained by removing the vertices in $U$ from the initial SEMs $(\B{\kappa}, \Di{\kappa})$ by $(\Bt{\kappa}, \Dt{\kappa})$, for $\kappa \in \Set{1,2}$ respectively. 
Then we have that $\Dt{1} = \Dt{2} = \Diag(\Set{\vart_j}_{j \in V'})$, and $\supp(\Bt{1} - \Bt{2}) = \supp(\B{1} - \B{2}) = \D^*$.

\textbf{Claim (ii):} The function \computeOrder{} returns a list of sets $O = (S_1, \ldots, S_m)$ such that for every $i \in S_a$ and $j \in S_b$ with $a < b$, we have $i \prec_{\tau} j$ for some $\tau \in \mathcal{T}(\D_G)$. 
The set of vertices in $R$ occurs before all the vertices in $O$ in the causal order $\tau$
for all $\tau \in \mathcal{T}(\D_G)$.

\textbf{Claim (iii):} For $O = (S_1, \ldots, S_m)$ and any $i,j \in S_a$ for $a \in [m]$, the nodes
$i$ and $j$ do not have an edge between them in $\D^*$.

\textbf{Claim (iv):} The function \orientEdges{} returns a $\D$ such that $\D \supseteq \D^*$.

\textbf{Claim (v):} The function \prune{} returns a $\D$ such that $\skel(\D) = \skel(\D^*)$. Further, if  $\Di{1} = \Di{2}$ then $\D = \D^*$.
\end{proof}
\paragraph{Computational Complexity.} Note that the \orientEdges{} step already removes quite a few edges from the difference DAG. Then in the \prune{} step, for those edges in the difference of the precision matrix that are incident on variables present in $O$,  we only test over subsets that are descendants of the nodes. 
Whereas, the method of \cite{wang_direct_2018} test for subsets over all $[p]\setminus U$ vertices for each edge in the difference of the precision matrix.
Thus, our method is \textit{strictly more efficient} than that of \cite{wang_direct_2018}. 
Further, the computational complexity of our algorithm lies between two extremes.
When the noise variances are the same, we have $R = \emset$ and our algorithm runs most efficiently. 
While, if all the noise variances are different, i.e.,  $\Di{1}_{i,i} \neq \Di{2}_{i,i}$ for all $i$, then $O = \emptyset$ while $R = V$ and the computational complexity of our algorithm is the same as \cite{wang_direct_2018} in terms of the number of tests performed during pruning.

\subsection{Finite-Sample Guarantees}
In this section, we derive finite sample guarantees for our algorithm. The performance of our method depends
on how accurately the difference between the precision matrices are estimated. The problem of directly estimating the difference between the precision matrices of two Gaussian SEMs (or more generally Markov Random Fields), given samples drawn from the two individual models, has received significant attention over the past few years \cite{zhao_direct_2014, belilovsky_testing_2016, yuan_differential_2017,
liu2017support}. Among these, the KLIEP algorithm of \cite{liu2017support} and the algorithm of \cite{zhao_direct_2014}
come with provable finite sample guarantees. We use the algorithm of \cite{zhao_direct_2014} for estimating the difference of precision matrices. 
Given \emph{sample covariance matrices} $\Smh{(1)}$ and $\Smh{(2)}$, \cite{zhao_direct_2014} estimate the difference of precision matrix by solving the following optimization problem:
\begin{align*}
\DOh &= \argmin_{\DO} \NormI{\DO} \text{ subject to }  \ \NormMax{\Smh{(1)}(\DO)\Smh{(1)} - \Smh{(2)} + \Smh{(1)}} \leq \lambda_n,
\end{align*}
where $\lambda_n$ is the regularization parameter and $\NormMax{\parg}$ denotes maximum absolute value of the matrix.
Denoting $\beta = \vectorize(\DO)$ and by using the properties of Kronecker product, the above optimization problem 
can be written as follows:
\begin{align}
\widehat{\beta} &= \argmin_{\beta} \NormI{\beta} \text{ subject to }  \ \NormMax{(\Smh{(2)} \otimes \Smh{(1)}) \beta - \vectorize(\Smh{(1)} - \Smh{(2)})} \leq \lambda_n. \label{eq:opt}
\end{align}

However, to decouple the analysis of our algorithm from those of the algorithms for estimating the difference between precision matrices, we state our results with respect to a finite-sample oracle for estimation of the difference between precision matrices. Specifically, we make the following assumption:
\begin{assumption}
\label{ass:oracle}
Given samples $\X{1} \in \R^{n_1 \times p}$ and $\X{2} \in \R^{n_2 \times p}$ drawn from two linear SEMs, there exists an estimator $\widehat{\D}_{\Omega}$ for the difference between the precision matrices such that
$\Prob[X_1, X_2]{\NormMax{\widehat{\D}_{\Omega} - \D_{\Omega}} \leq \varepsilon} \geq 1 - \delta$,
if $n_1 \geq \eta_1(\varepsilon, \delta)$ and $n_2 \geq \eta_2(\varepsilon, \delta)$, for some $\varepsilon, \delta > 0$ and functions $\eta_1, \eta_2$.
\end{assumption}
We will also need a finite sample version of the condition in Assumption \ref{ass:main} to obtain our
finite sample guarantees.
\begin{assumption}
\label{ass:finite}
Let $(\B{1}, \Di{1})$ and $(\B{2}, \Di{2})$ be two SEMs with the difference DAG given by $\D_G = ([p], \D)$,
where $\D = \supp(\B{1} - \B{2})$, and difference of precision matrix given by $\DO$.
Let $U = \Set{i \in [p] \mid  (\DO)_{i,*} = 0}$ and let $V = [p] \setminus U$.
Then, the two SEMs satisfy the following assumptions:
\begin{enumerate}[(i)]
\item  For $i \in U$, the edges and noise variances are invariant.
\item Let $\mathcal{S} = \Set{ \Set{\tau_1, \ldots, \tau_k} \subseteq [p] \mid \tau \in \mathcal{T}(\D_G), 0 \leq k \leq \Abs{V}}$.
For each $(i,j) \in \D$ and $\forall S \in \mathcal{S} : i,j \in S$, we have $\Abs{\corr{1}(X_i, X_j | X_{S'}) - \corr{2}(X_i, X_j | X_{S'})} \geq 2\varepsilon$, for $S' = S \setminus \Set{i,j}$ and for some $\varepsilon > 0.$
\end{enumerate}
\end{assumption}
\begin{remark}
\label{remark:epsilon}
The finite sample version of Algorithm \ref{alg:exact_alg} also takes as input a threshold $\varepsilon$, and thresholds the difference of precision matrices at $\varepsilon$. 
\end{remark}
The sample complexity of the finite sample
algorithm depends on the number of non-zero entries (edges) in the difference of precision matrix.
Throughout the course of our algorithm we compute difference of precision matrices over subsets $S \in \mathcal{S}$, where $\mathcal{S}$ is defined in Assumption \ref{ass:finite}. In what follows, $d$ denotes the number of non-zero entries in the densest difference of precision matrix, i.e., $d = \max_{S \in \mathcal{S}} \Norm{\D_{\Omega}^{S}}_0$.
 Our next theorems formally characterize the finite sample guarantees of our method given in Algorithm \ref{alg:exact_alg}.
\begin{theorem}
\label{thm:finite_sample}
Under Assumptions \ref{ass:oracle}, \ref{ass:finite}.
Let $\D_G = ([p], \D^*)$ be the true difference DAG, where $\D^* = \supp(\B{1} - \B{2})$.
Let $\X{\kappa} \in \R^{n_{\kappa} \times p}$ be samples drawn from the two SEMs, and let $\Smh{(\kappa)} = (\nicefrac{1}{n_{\kappa}}) \Tp{\X{\kappa}} \X{\kappa}$ be the sample covariance matrices for $\kappa \in \Set{1,2}$.  
Given $\Smh{(1)}$, $\Smh{(2)}$ and $\varepsilon > 0$ as input, the finite sample version of Algorithm \ref{alg:exact_alg} returns $\D$ such that $\skel(\D) = \skel(\D^*)$ and all the directed edges in $\D$ are correctly oriented with probability at least $1 - \delta$ if $n_1 \geq \eta_1(\varepsilon, \delta)$ and $n_2 \geq \eta_2(\varepsilon, \delta)$ for some $\delta > 0$.
Furthermore, if $\Di{1} = \Di{2}$ then $\D = \D^*$.
\end{theorem}

\begin{theorem}[Adapted from \cite{zhao_direct_2014}]
\label{thm:cai_precision_difference}
Let $\Sm{\kappa}$ denote the true covariance matrices of the two SEMs respectively, for $\kappa \in \Set{1,2}$.
Also, let $\X{\kappa} \in \R^{n_{\kappa} \times p}$ be samples drawn from the two SEMs, for $\kappa \in \Set{1,2}$. 
Let $\DOh$ be the estimate of the difference of precision matrix obtained by solving \eqref{eq:opt}.
Define $\komax \defeq \max_{(i,j) \neq (k,l)}  \Abs{ \Sm{1}_{i,j} \Sm{2}_{k,l}}$
and $\kdmin \defeq \min_{i} \Sm{1}_{i,i} \Sm{2}_{i,i}$. 
Let $\eigmin(\parg)$ denote the minimum eigenvalue of a matrix. 
If $\komax \leq \frac{\eigmin(\Sm{1}) \eigmin(\Sm{2})}{2 \Norm{\DO}_0}$,
the regularization parameter, $\lambda_n$, and the number of samples, $n$, satisfy the following conditions:
\begin{align*}
n \geq \frac{C^2}{(\kdmin \varepsilon)^2} \log \frac{2p}{\delta} \quad \text{ and } \quad
\lambda_n \geq C \sqrt{\frac{1}{n} \log \frac{2p}{\delta}},
\end{align*}
where $C$ is a constant that depends linearly on 
$\Abs{\DO}_1$, $\NormMax{\Sm{\kappa}}$, and $\max_{\kappa, i} \Sm{\kappa}_{i,i}$,
then with probability at least $1 - \delta$ we have that $\NormMax{\DO - \DOh} \leq \varepsilon$.
\end{theorem}
For the proof of the above theorem, given in Appendix \ref{app:proofs}, we 
 adapt the proof of \cite{zhao_direct_2014} to obtain finite sample results in the form required
by Theorem \ref{thm:finite_sample}. Specifically, we analyze the optimization problem given by \eqref{eq:opt}, whereas
\cite{zhao_direct_2014} only estimate the upper diagonal of the difference of precision matrix $\DO$ thereby improving the computational complexity of estimation at the cost of requiring more stringent conditions on the true covariance matrices. 
We also use concentration of covariance matrix results from \cite{ravikumar2011high} to obtain finite sample results in the form required by Theorem \ref{thm:finite_sample}. 
From the above theorem we can conclude that the method of \cite{zhao_direct_2014} requires an incoherence condition on the true covariance matrices --- which is similar to known incoherence conditions for estimating precision matrices \cite{ravikumar2011high} --- for direct estimation of the difference  of precision matrices. 
Furthermore, the true difference of precision matrix essentially needs to have constant sparsity, i.e., the number of non-zero entries in the difference of precision matrix ($\Norm{\DO}_0)$ should be constant in the high-dimensional regime for the constant $C$ in the above theorem to not depend on $p$.
Finally, we then have the following finite sample result on estimating the difference DAG using \eqref{eq:opt}.

\begin{corollary}
\label{cor:finite_sample}
Using \eqref{eq:opt} to estimate the difference of precision matrices, if $\min(n_1, n_2) = 
\BigO{(\frac{d^2}{\varepsilon^2}) \log (\frac{p}{\delta})}$, $\komax \leq \frac{\eigmin(\Sm{1}) \eigmin(\Sm{2})}{2 d}$,
and $\lambda_n = \BigOm{\sqrt{\frac{1}{n} \log \frac{2p}{\delta}}}$, where the constant $\komax$ is defined in 
Theorem \ref{thm:cai_precision_difference}, and the true difference DAG satisfies Assumption \ref{ass:finite},
then the following holds with probability at least $1 - \delta$.
The finite sample version of  Algorithm \ref{alg:exact_alg} returns $\D$ such that $\D = \D^*$  when
$\Di{1} = \Di{2}$, otherwise $\skel(\D) = \skel(\D^*)$ and all the directed edges in $\D$ are correctly oriented.
\end{corollary}
The following proposition compares the sample complexity of our method against indirect methods that
first learn the individual SEMs and then compute the difference. 
The detailed comparison can be found in Appendix  \ref{app:comparison}.
\begin{proposition}
If the incoherence condition $\komax \leq \nicefrac{\eigmin(\Sm{1}) \eigmin(\Sm{2})}{2 \Norm{\DO}_0}$ is satisfied, then the sample
complexity of our method is strictly better than indirect estimation methods  \cite{ghoshal2018learning}, with the 
former only depending on $\Abs{\DO}_1$ while the latter depending on $\NormI{\Om{\kappa}}^2$.
\end{proposition}
Lastly, in the absence of any sample complexity guarantees for DCI \cite{wang_direct_2018}, it is difficult to make any theoretical comparison with it. However, given empirical results on the superiority of direct estimation methods for precision matrices \cite{danaher2014joint}, and that in the worst case when the noise variances are different, our method performs the same number of tests as \cite{wang_direct_2018} without computing individual regression coefficients, we believe our method to have strictly better sample complexity than \cite{wang_direct_2018}.

\subsection{Fundamental Limits}
In this section, we obtain fundamental limits on the sample complexity of direct estimation of the
difference DAG. Towards that end, we consider the minimax error of estimation which we define over the subsequent lines.
Let $\Xf = (\X{1}, \X{2})$ be the two sets of $n$ samples generated from the product distribution $P^n = P_1^n \times P_2^n$
where $P_{\kappa}$ corresponds to a Gaussian linear SEM for $\kappa \in \Set{1, 2}$. Let $\Pf$ be the family of all such 
product distributions such that the DAGs $\G{1}$ and $\G{2}$ share the same causal order.
 We will denote the corresponding DAG for the distribution $P_{\kappa}$ by $G(P_{\kappa})$ and we will denote the difference 
DAG by $\D_G(P)$. Let $\Dec$ be a decoder that takes as input the two sets of samples $\Xf$ and returns a difference DAG $\Dec(\Xf)$.
The minimax estimation error is then defined as:
\begin{align}
\perr \defeq \inf_{\Dec} \sup_{P \in \Pf} \Prob[\Xf \sim P^n]{\Dec(\Xf) \neq \D_G(P)}, \label{eq:perr}
\end{align}
where the infimum is taken over all decoders that take as input two sets of samples drawn from a distribution $P \in \Pf$
and return a difference graph.  The following theorem lower bounds the minimax error.
\begin{theorem}
\label{thm:lower_bound}
Given $n$ samples drawn from each of the two linear Gaussian SEMs with DAGs $\G{1}$ and $\G{2}$ such that the DAGs share the same causal order, and the difference DAG $\D_G$ is sparse with each node having at most $\dt$ parents.  
If the number of samples $n \leq (\nicefrac{\dt}{2}) \log (\nicefrac{p}{2\dt}) - (\nicefrac{2}{p}) \log 2$ then $\perr \geq \nicefrac{1}{2}$,
where $\perr$ is defined in \eqref{eq:perr}.
\end{theorem}
\begin{remark}
To the best of our knowledge, Theorem \ref{thm:lower_bound} is the first information-theoretic lower bound to the minimax error of difference-DAG estimation.
Our result shed lights on the necessary number of samples for learning the difference DAG under \textbf{any} method, and shows the logarithmic dependence on the number of variables.
\end{remark}
Next, we compare the sample complexity of our method against indirect methods that first estimate the individual SEMs and then compute the difference.
\section{Experiments}
\label{sec:exps}
In this section, we describe empirical results from running our Algorithm \ref{alg:exact_alg} on synthetic data with the goal of verifying our theoretical contributions.
For the population case, we test graphs of size up to $p=100$ nodes, while for the finite-sample case, we test graphs of small size due to high computational cost of indirect methods.
We also provide a real-world experiment in Appendix \ref{app:real_world_exp}, where the number of nodes is $p=157$.

\paragraph{Generative process.} 
For generating a random SEM pair, we first generate an \Erdos-\Renyi random DAG on $p$ nodes with average neighborhood size $N_a$. 
Then, we generate the second DAG by deleting an existing edge or adding a new edge, consistent with the topological ordering of the first DAG, with probability $0.05$ each.
Thus, each DAG has an average of $\nicefrac{p N_a}{2}$ edges out of $\nicefrac{p(p-1)}{2}$ possible edges.
We set the edge weights to be uniformly at random in the set $[-1, -0.25] \union [0.25, 1]$, while noise variances are set to follow a normal distribution $\mathcal{N}(0,1)$.

\paragraph{Population setting.}
As dictated by Theorem \ref{thm:population},  Algorithm \ref{alg:exact_alg} should return the exact structural difference of SEMs if given the population covariance matrix of each SEM.
In Table \ref{tab:population}, we show that on graphs of up to $p=100$ nodes, in effect, Algorithm \ref{alg:exact_alg} returns the true difference-DAG.
\begin{table}[!tb]
\centering
\caption{Following our generative process described above, we generated $50$ pairs of networks with $p=\{10, 50, 100\}$ variables with expected neighborhood size of $N_a = \{6, 35, 80\}$, respectively.
That is, for each of those values of $p$, each network in the pair of DAGs is \textbf{dense} with expected number of $\{30, 875, 4000\}$ edges, respectively.
We obtained that $\Delta = \Delta^*$ almost surely, after using our Algorithm \ref{alg:exact_alg} and the population covariance matrices as inputs.}
\label{tab:population}
\begin{tabular}{@{}cccc@{}}
\toprule
\textbf{$p$} & \textbf{$N_a$} & {Expected $|E|$ of each DAG} & $\hat{\mathrm{Pr}} \{ \Delta = \Delta^* \}$ \\
\midrule
10  & 6  & 30   & 1.0 \\
50  & 35 & 875  & 1.0 \\
100 & 80 & 4000 & 1.0 \\
\bottomrule
\end{tabular}
\end{table}

\paragraph{Finite-sample setting.}
For experiments with a finite number of samples, we also follow our generative process above.
We generate $50$ pairs of DAGs with $p=10$ and $N_a=6$, and make sure that the pair of SEMs have $\varepsilon \geq 0.1$ in Assumption \ref{ass:finite}.
We then generate $\floor{e^{C} \log p}$ number of samples from each SEM for $C \in [5, 13]$.
In Figure \ref{fig:plot}, we compare against the algorithms: PC \cite{spirtes2000causation}, GES \cite{meek1997graphical}, MMHC \cite{tsamardinos2006max}, and LiSTEN \cite{ghoshal2018learning}, all of which first learn each SEM \textbf{separately} and then output the difference of adjacency matrices as the difference DAG.
Finally, we also compare against the DCI-C method \cite{wang_direct_2018}, which, as in our setting, also estimates the difference of SEMs.
We note how traditional state-of-the-art methods (PC, GES, MMHC, LiSTEN) suffer learning the difference DAG as each DAG independently is \textbf{dense}.
The closest to our results is the DCI-C method, although, as seen in Figure \ref{fig:plot} (Left), our algorithm performs better in the small sample complexity regime.
\begin{figure}[!t]
\centering
\includegraphics[width=0.49\linewidth]{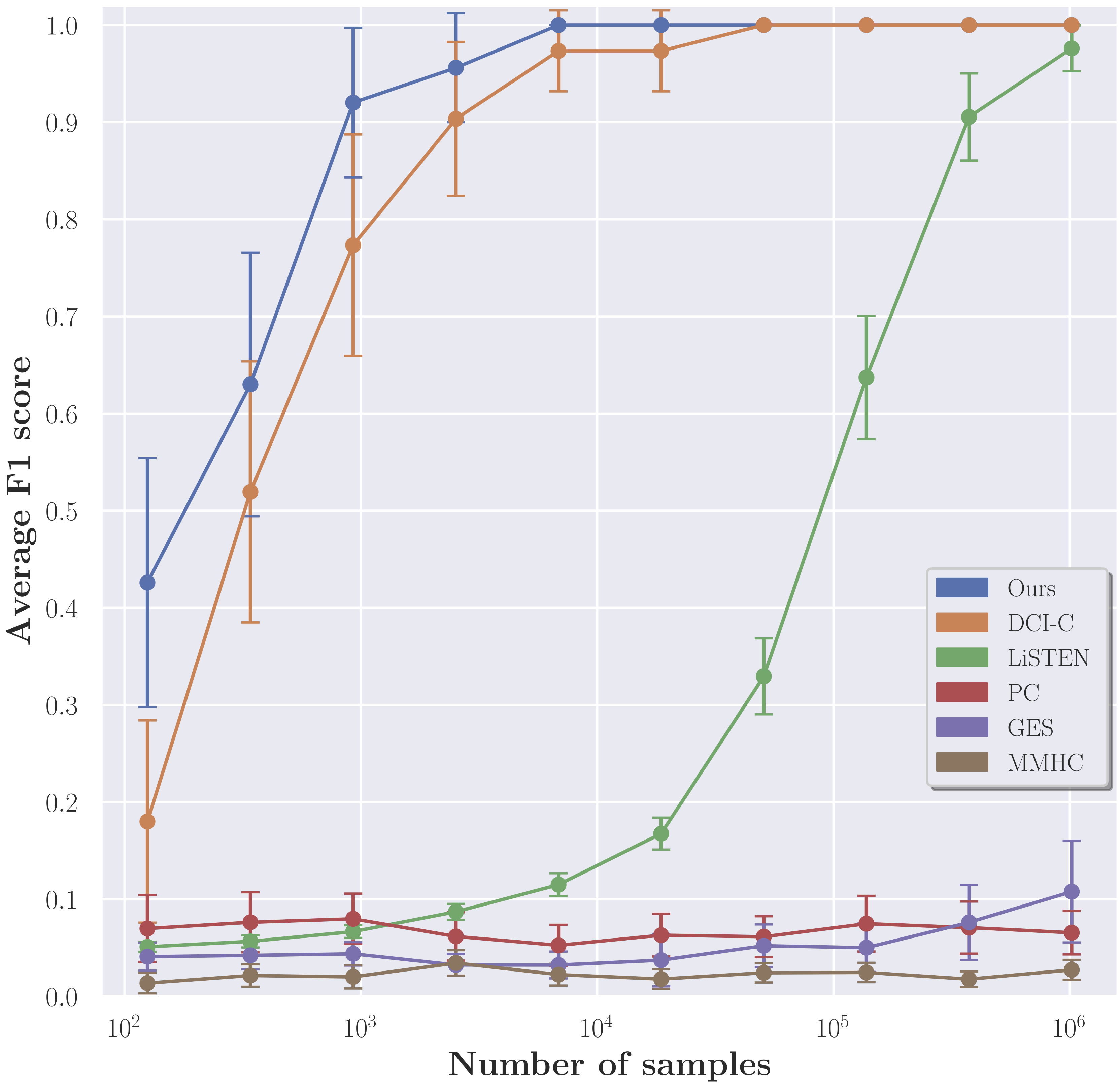} \hfill
\includegraphics[width=0.49\linewidth]{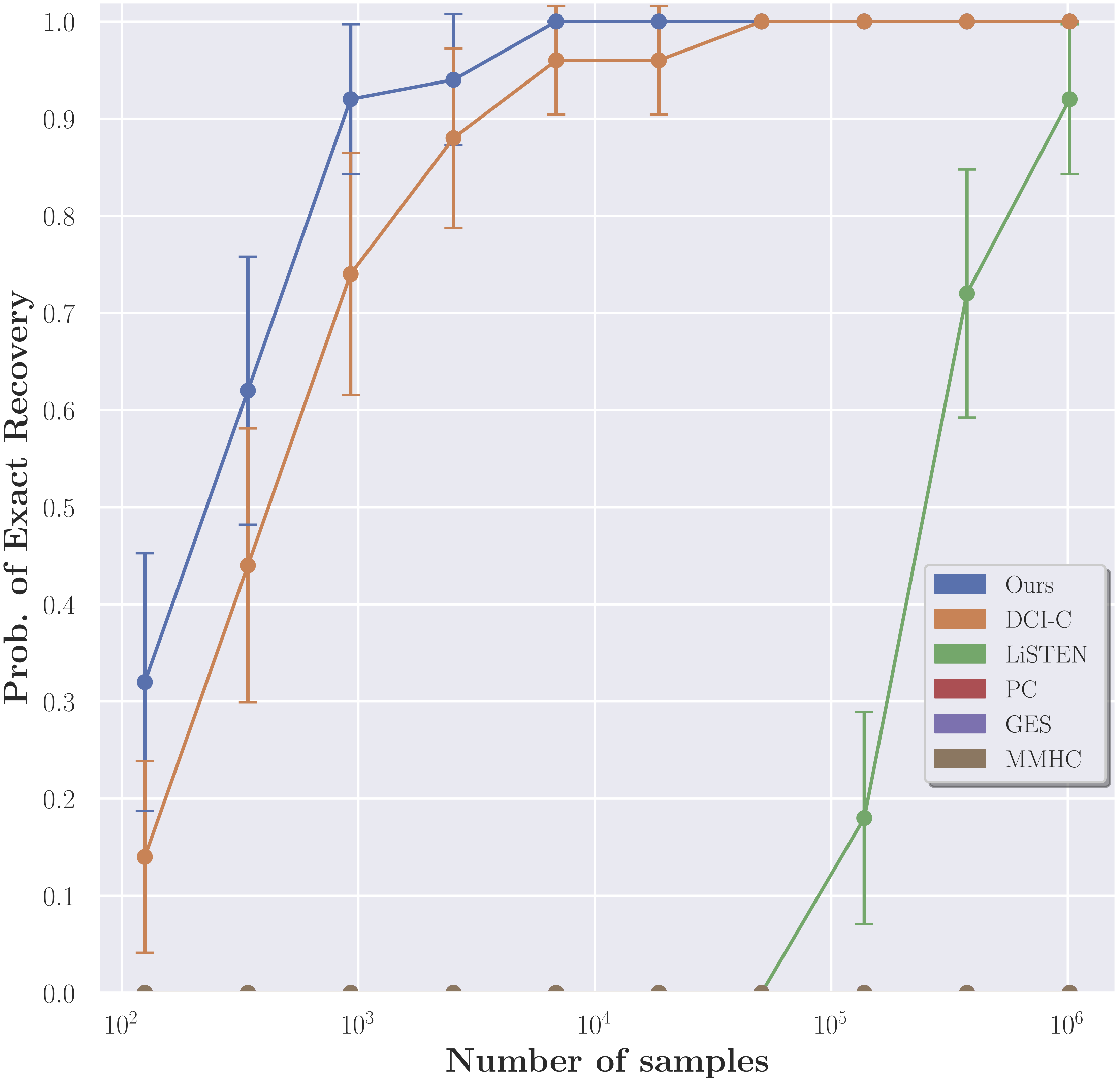}
\caption{(Left) Average of F1 scores and (Right) proportion of exact recovery, computed across 50 repetitions.}
\label{fig:plot}
\end{figure}
\begin{remark}
We emphasize that the reason to set $p = 10$ in the above experiment is due to the exponential computational cost of the PC algorithm for learning each \textbf{dense} SEM separately.
\end{remark}

\paragraph{Unequal noise variance.}
We now set out to understand the performance of our algorithm when we set different noise variances.
As shown above, the DCI-C algorithm \citep{wang_direct_2018} is the most comparable method to ours, thus, here we compare against it.
For this experiment, we sampled pairs of SEMs under the same finite-sample setting.
However, instead of fixing the noise variance to be one for all nodes, we set the noise variance for each node to be one of $\{1-\gamma, 1, \gamma\}$ with probability $\nicefrac{1}{3}$, where $\gamma$ is a noise parameter. 
In Figure \ref{fig:gammas}, we note that we still achieve close-to-perfect recovery in a stable manner, while the performance of DCI-C decreases as the change in noise variance increases.
\begin{figure}[!t]
\centering
\includegraphics[width=0.49\linewidth]{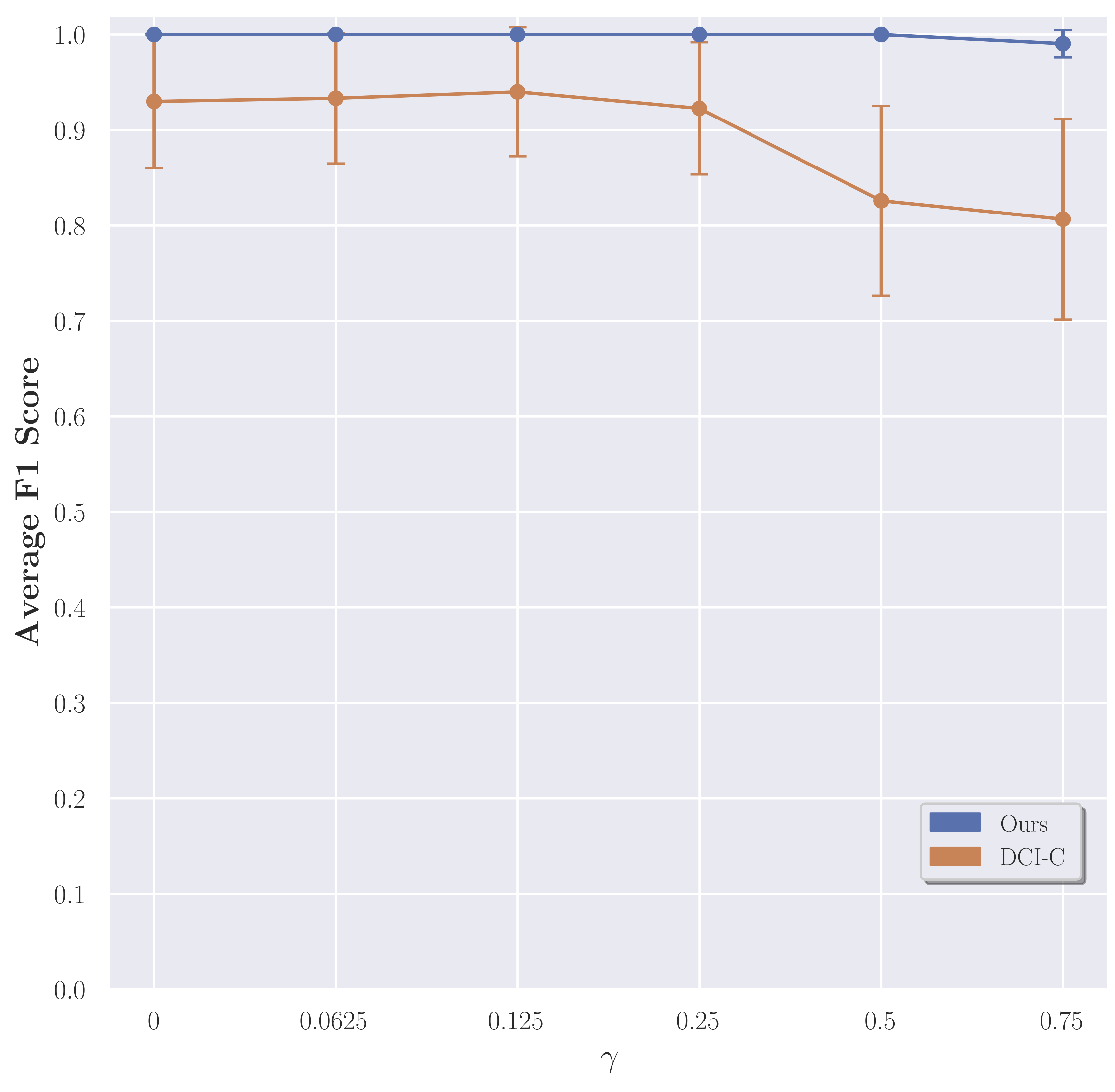} \hfill 
\includegraphics[width=0.49\linewidth]{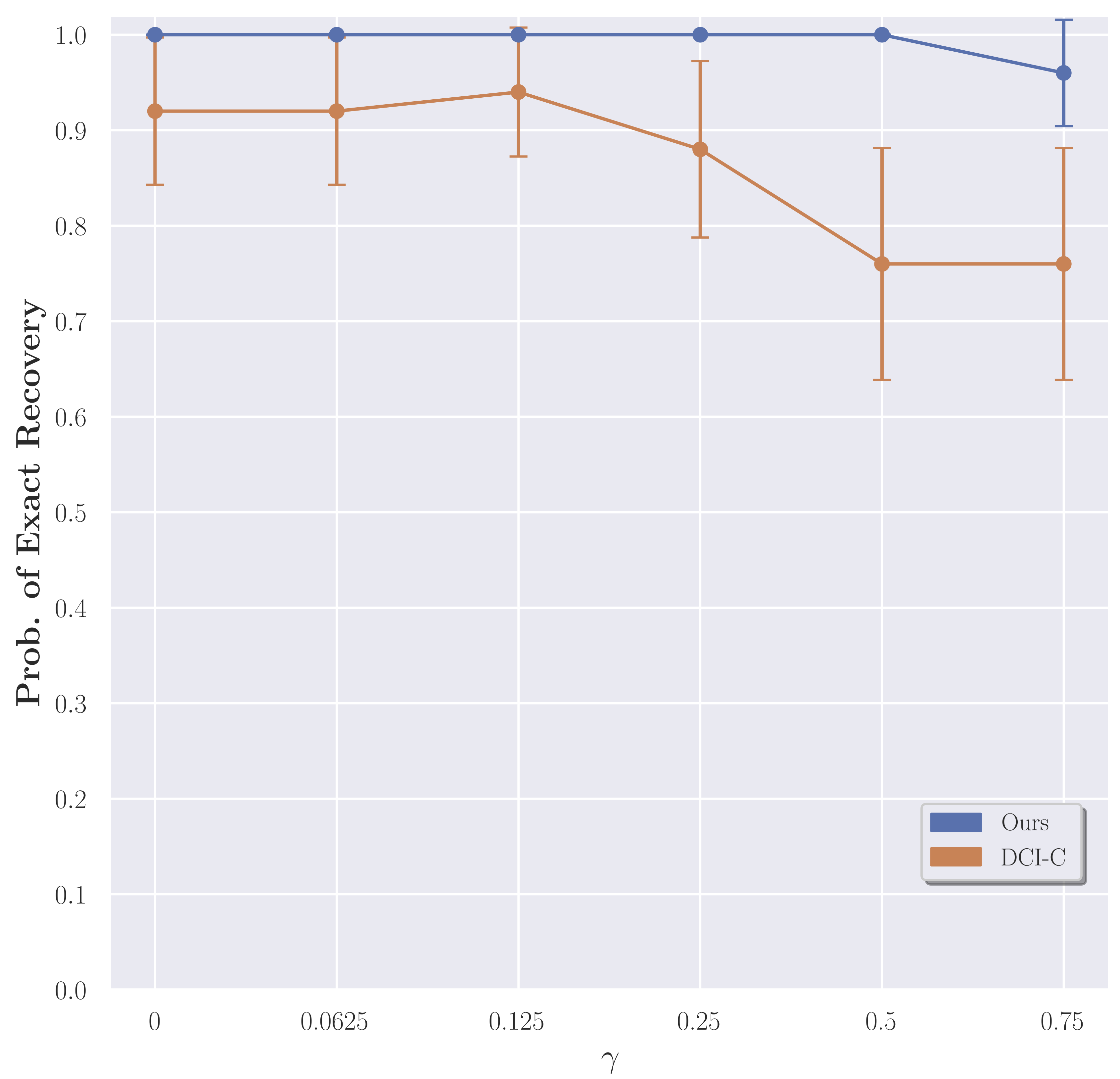}
\caption{(Left) Average of F1 scores, and (Right) Prob. of exact recovery, computed across 50 repetitions.
In this case the number of samples was set to $C=9$ (see finite-sample setting).}
\label{fig:gammas}
\end{figure}

\paragraph{Additional experiments.} To conclude our empirical results, we present two additional experiments in the appendix.
In Appendix \ref{app:log_phase}, we empirically corroborate the logarithmic dependence on the number of variables, while in Appendix \ref{app:real_world_exp}, we show an experiment on a real-world dataset with $p=157$ variables from the medical domain, which demonstrates the applicability of our method.

\section{Conclusion}
In this paper we considered the problem of directly estimating the difference-DAG of two linear SEMs, that share the same causal order, from samples generated from the individual SEMs. 
We showed that if the number of samples from each SEM grows as $\BigO{d^2 \log p}$ where $d$ is the number of edges in the (densest) difference of moral sub-graphs, 
and under an incoherence condition on the true covariance and precision matrices, our algorithm recovers either the correct DAG or partially directed DAG with the correct
skeleton and correct orientation of directed edges depending on whether or not the noise variances are the same.
We also showed that any algorithm requires $\Omega(\dt \log (p/\dt))$ samples to estimate the difference DAG consistently where $\dt$ is the maximum number of parents of a node in the difference DAG.

\normalsize
\bibliographystyle{agsm}
\bibliography{paper}

\normalsize
\clearpage
\onecolumn
\begin{center}
\Large{\bf Supplementary Material}\\
\large{\bf Direct Learning with Guarantees of the Difference DAG Between \\Structural Equation Models}
\end{center}
\hrule
\begin{appendix}
\section{Detailed proofs}
\label{app:proofs}

\begin{proof}[Proof of Lemma \ref{lemma:delete}]
From Proposition 3 of \cite{ghoshal2018learning} we have that for a terminal vertex $j$: $\nicefrac{1}{\var_j} = \Omega_{j,j}$.
Thus for an arbitrary vertex $j$, the inverse of the noise variance of $j$ is given by the corresponding diagonal entry
of the precision matrix obtained by removing all descendants of $j$. Further, the precision matrix over a subset of
vertex $S \leq [p]$ is given by the Schur-complement $\Omega_{S, S} - \Omega_{S, S^c} (\Omega_{S^c, S^c})^{-1} \Omega_{S^c,S}$,
where $S^c$ denotes the complement of $S$. Note that $\Anct_j = \Anc_j \setminus U_j$, $\Anct_j^c = \Anc_j^c \union U_j$,
and $\Anc_j^c \intersection U_j = \emset$.
Therefore, 
\begin{align*}
\frac{1}{\vart_j} &= \Omega_{j,j} - (\Omega_{j,\Anct_j^c}) (\Omega_{\Anct_j^c,\Anct_j^c})^{-1} (\Omega_{\Anct_j^c,j}) \\
&= \Omega_{j,j} - (\Omega_{j,\Anc_j^c}, \Omega_{j,U_j}) \matrx{\Omega_{\Anc^c_j,\Anc^c_j} & \Omega_{\Anc^c_j,U_j} \\
\Omega_{U_j, \Anc^c_j} & \Omega_{U_j, U_j}}^{-1} \matrx{\Omega_{U_j,j} \\ \Omega_{\Anc_j^c,j}} \\
&= \Omega_{j,j} - (\Omega_{j,\Anc_j^c}, \Omega_{j,U_j}) 
\matrx{\Omega_{\Anc_j^c,\Anc_j^c}^{-1} + P & Q \\
\Tp{Q} & R} \matrx{\Omega_{U_j,j} \\ \Omega_{\Anc_j^c,j}},
\end{align*}
where
\begin{align*}
P &\defeq  \Omega_{\Anc_j^c,\Anc_j^c}^{-1} \Omega_{\Anc^c_j, U_j} R \Omega_{U_j, \Anc^c_j} \Omega_{\Anc_j^c,\Anc_j^c}^{-1}, \\
Q &\defeq - \Omega_{\Anc_j^c,\Anc_j^c}^{-1} \Omega_{\Anc^c_j, U_j} R, \\
R &\defeq = (\Omega_{U_j, U_j} - \Omega_{U_j, \Anc^c_j} \Omega_{\Anc_j^c,\Anc_j^c}^{-1} \Omega_{\Anc^c_j, U_j})^{-1}.
\end{align*}
Now, writing 
\begin{align*}
\matrx{\Omega_{\Anc_j^c,\Anc_j^c}^{-1} + P & Q \\
\Tp{Q} & R} = \matrx{\Omega_{\Anc_j^c,\Anc_j^c}^{-1} & 0 \\ 0 & 0} + 
\matrx{P & Q \\ \Tp{Q} & R},
\end{align*}
and some algebraic manipulations later we have that:
\begin{align*}
\frac{1}{\vart_j} &= \frac{1}{\var_j} - (\Omega_{j, \Anc^c_j} \Omega_{\Anc_j^c,\Anc_j^c}^{-1} \Omega_{\Anc^c_j, U_j} - \Omega_{j,U_j}) R \times \\
	&\qquad (\Omega_{U_j, \Anc^c_j} \Omega_{\Anc_j^c,\Anc_j^c}^{-1} \Omega_{\Anc^c_j, U_j} - \Omega_{U_j,j}) \\
&= \frac{1}{\var_j}  - (\Omega^{\Anc_j}_{j, U_j})(\Omega^{\Anc_j}_{U_j, U_j})^{-1}(\Omega^{\Anc_j}_{U_j,j}) \\
& = \frac{1}{\var_j}  - (B_{j, U_j})(\Omega^{\Anc_j}_{U_j, U_j})^{-1}\Tp{(B_{j, U_j})},
\end{align*}
where the last line follows from Proposition 4 of \cite{ghoshal2018learning} since $j$ is a
terminal vertex in the induced subgraph over $\Anc_j$. For characterizing the edge weights,
observe once again that $j$ is a terminal vertex in the induced subgraph over $\Anc_j$, and therefore from
from Proposition 4 of \cite{ghoshal2018learning} we have that:
\begin{align*}
\wtl{B}_{j,k} = -\vart_j (\Omega_{j,k} - (\Omega_{j,\Anct_j^c}) (\Omega_{\Anct_j^c,\Anct_j^c})^{-1} (\Omega_{\Anct_j^c, k})).
\end{align*}
The final result follows from following the previously derived steps for the noise variance.
\end{proof}
\begin{proof}[Proof of Theorem \ref{thm:population}]
Let $V' = [p] \setminus U$ and let $p' = \Abs{V'}$, where $U$ is the set of invariant vertices defined
in line \ref{line:U} of the main algorithm. Denote the two initial DAGs by $\G{1}$ and $\G{2}$.
Let $\mathcal{T}(\D_G)$ be the set of topological orderings
induced by the true difference DAG. The correctness of Algorithm \ref{alg:exact_alg} 
follows from the following claims, which  we prove subsequently. 
\begin{itemize}
\item \emph{Claim (i)}: 
Denote the SEMs obtained by removing the vertices in $U$ from the initial SEMs $(\B{\kappa}, \Di{\kappa})$ by
$(\Bt{\kappa}, \Dt{\kappa})$, for $\kappa \in \Set{1,2}$ respectively. Then 
we have that $\Dt{1} = \Dt{2} = \Diag(\Set{\vart_j}_{j \in V'})$, 
and $\supp(\Bt{1} - \Bt{2}) = \supp(\B{1} - \B{2}) = \D^*$.

\item \emph{Claim (ii)}: The function \computeOrder{} returns a list of sets $O = (S_1, \ldots, S_m)$
such that for every $i \in S_a$ and $j \in S_b$ with $a < b$,  we have $i \prec_{\tau} j$ for some $\tau \in \mathcal{T}(\D_G)$.

\item \emph{Claim (iii)}: For $O = (S_1, \ldots, S_m)$ and any $i,j \in S_a$ for $a \in [m]$, the nodes
$i$ and $j$ do not have an edge between them in $\D^*$.

\item \emph{Claim (iv)}: The function \orientEdges{} returns a $\D$ such that $\D \supseteq \D^*$.

\item \emph{Claim (v)}: The function \prune{} returns a $\D$ such that $\D = \D^*$.
\end{itemize}

\paragraph{Proof of Claim (i).} Lemma \ref{lemma:delete} gives the characterization $(\Bt{\kappa}, \Dt{\kappa})$
for $\kappa \in \Set{1,2}$. By Assumption \ref{ass:main} (i) we have that for each $j \in U$,
$\B{1}_{j,k} = \B{2}_{j,k}$ $\forall k$. Note that by definition of $U$, for any $u, v \in U$,
we have that $\Om{1}_{u,v} = \Om{2}_{u,v}$. Next we will show that for any node $j \in V'$,
$\Om{1, \Anc_j}_{u,v} = \Om{2, \Anc_j}_{u,v}$ for any $u,v \in U_j$ (recall that $U_j = U \intersection \Anc_j$
and $\Anc_j$ is the set of ancestors of $j$ in the toplogical order in the initial SEMs).
Since $\Om{\kappa, \Anc_j}$ is the precision matrix for the SEM
obtained by removing $\Anc_j^c$, we have that for any $u, v \in U_j$, $\Anc_j^c$ contains vertices that occur
after $u$ and $v$ in the causal order. Therefore by Lemma \ref{lemma:delete}:
\begin{align*}
\Om{1,\Anc_j}_{u,v} &= -\nicefrac{\B{1}_{u,v}}{\var_u} -\nicefrac{\B{1}_{v,u}}{\var_v} + 
\sum_{l \in \Chi{1}(u) \intersection \Chi{1}(v) \intersection \lAnc{1}_j} \mkern-18mu \nicefrac{\B{1}_{l,u} \B{1}_{l,v}}{\var_l} \\
&= -\nicefrac{\B{2}_{u,v}}{\var_u} -\nicefrac{\B{2}_{v,u}}{\var_v} + 
 \sum_{l \in \Chi{2}(u) \intersection \Chi{2}(v) \intersection \lAnc{2}_j} \mkern-18mu \nicefrac{\B{2}_{l,u} \B{2}_{l,v}}{\var_l} 
= \Om{2,\Anc_j}_{u,v}.
\end{align*} 
Therefore, once again by Lemma \ref{lemma:delete} we have that $\Om{1,\Anc_j} = \Om{2,\Anc_j}$ and thus $\smt{1}_j = \smt{2}_j = \wtl{\sigma}_j$,
and $\Bt{1}_{j,k} = \Bt{2}_{j,k}$, $\forall j,k \notin U$, where $\wtl{\sigma}_j$ and $\Bt{\kappa}_{j,k}$ are given by Lemma \ref{lemma:delete}.
Thus we have that $\supp(\Bt{1} - \Bt{2}) = \supp(\B{1} - \B{2}) = \D^*$.

From this point onwards all the arguments will be w.r.t. the two SEMs $(\Bt{\kappa}, \wtl{D})$ and thus $\DO$ will denote
the difference of precision matrix over $(\Bt{1}, \wtl{D})$ and $(\Bt{2}, \wtl{D})$ having DAGs $\Gt{1}$ and $\Gt{2}$.

\paragraph{Proof of Claim (ii).} From Assumption \ref{ass:main}(i)
and Proposition \ref{prop:terminal_vertex} we have that $i$ is a terminal
vertex in $\D_G$ if and only if $(\DO)_{i,i} = 0$. From Lemma \ref{lemma:delete} we have that
removing a set of vertices does not change the topological ordering, i.e., 
$\mathcal{T}(\Gt{\kappa}) \subseteq \mathcal{T}(\G{\kappa})$, for $\kappa \in \Set{1,2}$. Therefore,
the order in which the vertices are eliminated by the function \computeOrder{} is consistent with the
topological order of the difference DAG.

\paragraph{Proof of Claim (iii).} For any two vertices $i,j$ such that $i,j \in S_a$ for $a \in [m]$, that
means $i$ and $j$ were eliminated in the same iteration and $(\DO)_{i,i} = (\DO)_{j,j} = 0$. However, if 
$(i,j) \in \D^*$ then by Assumption \ref{ass:main}(ii) $(\DO)_{j,j} \neq 0$. Therefore, $(i,j) \notin \D^*$.

\paragraph{Proof of Claim (iv).} From Assumption \ref{ass:main}(ii) we have that for any $(i,j) \in \D^*$,
$(\DO)_{i,j} \neq 0$. Also by Claim (ii) we have that the ordering $O$ is consistent with the topological 
ordering of the difference DAG $\D_G$. Therefore, we have that $\D \supseteq \D^*$.

\paragraph{Proof of Claim (v).} Let $\D'$ denote the set of edges returned by \orientEdges{}. For any $(i,j) \in \D' \setminus \D^*$
we have that $(\DO)_{i,j} \neq 0$ and $i,j \notin S_a$ for some $S_a \in O$. Then we have that
\begin{align*}
(\DO)_{i,j} = \sum_{l \in CC^{(1)}_{i,j}} \nicefrac{1}{\vart_l}({\Bt{1}_{l,i} \Bt{1}_{l,j}}) - 
	\sum_{l \in CC^{(2)}_{i,j}} \nicefrac{1}{\vart_l}({\Bt{2}_{l,i} \Bt{2}_{l,j}})
\end{align*}
where $CC^{(\kappa)}_{i,j} = \Chi{\Gt{\kappa}}(i) \intersection  \Chi{\Gt{\kappa}}(j)$
is the set of common children of $i$ and $j$ in the SEM indexed by $\kappa$.
Therefore, if we remove the nodes $CC^{(1)}_{i,j} \union CC^{(2)}_{i,j}$ from $V'$ and
compute the difference of precision matrix over $S = V' \setminus CC^{(1)}_{i,j} \union CC^{(2)}_{i,j}$,
then by Lemma \ref{lemma:delete} $(\D_{\Omega}^{S})_{i,j} = 0$. Since the nodes $CC^{(1)}_{i,j} \union CC^{(2)}_{i,j}$ are descendants
of $i$ and $j$ in the difference DAG, the function \prune{} will correctly remove the edge $(i,j)$. Thus 
if $\D$ is the set returned by \prune{} then $\D = \D^*$.
\end{proof}
\begin{proof}[Proof of Theorem \ref{thm:finite_sample}]
Note that by Assumption \ref{ass:oracle}, $\NormMax{\widehat{\D}_{\Omega}^S - \D_{\Omega}^S} \leq \varepsilon$ holds
with probability at least $1 - \delta$ simultaneously over all subsets $S \subseteq [p]$. 
Therefore, in the finite sample version given an $\varepsilon$-accurate estimate of $\DO$ and thresholding $\DO$ at $\varepsilon$, we
have that each line involving $\DO$ holds (by Assumption \ref{ass:finite}) with probability at least $1 - \delta$ simultaneously. So the claim follows.
\end{proof}
\begin{proof}[Proof of Theorem \ref{thm:cai_precision_difference}]
For the purpose of the proof, symbols superscripted by $*$ will correspond to ``true'' objects (e.g. true covariance matrix),
while symbols with a hat will denote the corresponding finite sample estimates.
Let $\Smh{} = \Smh{(1)} \otimes \Smh{(2)}$, $\Sigma^* = \Sm{1} \otimes \Sm{2}$,
$\bh = \vectorize(\Smh{1} - \Smh{2})$ and $\beta^* = \vectorize(\DOt)$,
where $\DOt$ is the true difference of precision matrices.
Then,
\begin{gather}
[\Sigma^*(\betah - \beta^*)]_i = \sum_{j} \Sigma^*_{i,j} (\betah_j - \beta^*_j) = \Sigma^*_{i,i} (\betah_i - \beta^*_i) 
	+ \sum_{j \neq i} \Sigma^*_{i,j} (\betah_j - \beta^*_j) \notag \notag \\
\implies \Abs{[\Sigma^*(\betah - \beta^*)]_i - \Sigma^*_{i,i} (\betah_i - \beta^*_i)}
	= \Abs{\sum_{j \neq i} \Sigma^*_{i,j} (\betah_j - \beta^*_j)} \qquad (\forall i \in [p])	\notag \\
\implies \kdmin \Abs{\betah_i - \beta^*_i} \Ineq[(a)]{\leq}
	\Abs{\sum_{j \neq i} \Sigma^*_{i,j} (\betah_j - \beta^*_j)} + \Abs{[\Sigma^*(\betah - \beta^*)]_i} \qquad (\forall i \in [p]) \notag \\
\implies \kdmin \NormInfty{\betah - \beta^*} \Ineq[(b)]{\leq}
	\komax \NormI{\betah - \beta^*} + \NormInfty{\Sigma^*(\betah - \beta^*)}, \label{eq:fsp_main_ineq}
\end{gather}
where (a) follows from reverse triangle inequality and (b) follows from taking max over $i$.
To upper bound $\NormInfty{\betah - \beta^*}$ we need to upper bound $\NormI{\betah - \beta^*}$ and $\NormInfty{\Sigma^*(\betah - \beta^*)}$.
Next, we will upper bound $\NormI{\betah - \beta^*}$.

To upper bound $\NormI{\betah - \beta^*}$ \textbf{assume that $\beta^*$ is feasible (for which we will provide a proof at the end).}
Thus we have that $\NormI{\betah} \leq \NormI{\beta^*}$. Let $S$ be the support of $\beta^*$, i.e., 
$S \defeq \Set{i \in [p] \mid \beta^*_{i} \neq 0}$. Let $\Sc$ be the complement of the set $S$.
Then
\begin{gather}
\NormI{\betah_S} + \NormI{\betah_{\Sc}} \leq \NormI{\beta^*_S} \notag \\
\implies \NormI{\betah_{\Sc}} \leq \NormI{\beta^*_S} - \NormI{\betah_S} \leq \NormI{\beta^*_S - \betah_S} \label{eq:fsp_l1_ineq1}
\end{gather}
From the above and the fact that $\beta^*_{\Sc} = 0$ we have
\begin{gather}
\implies \NormI{\beta^* - \betah} \leq \NormI{\beta^*_S - \betah_S} + \NormI{\beta^*_{\Sc} - \betah_{\Sc}} \leq 
	2 \NormI{\beta^*_S - \betah_S} \label{eq:fsp_l1_ineq2}
\end{gather}
Next, let $x = \betah - \beta^*$ and $\xb_S = (\xb_i)_{i \in [p]}$ such that
\begin{align*}
\xb_i = \left \{ 
\begin{array}{cc}
x_i & i \in S \\
0 & \mathrm{otherwise}
\end{array}
\right. .
\end{align*}
Next we have that
\begin{align}
\Abs{\Tp{\xb_S} \Sigma^* x} &= \Abs{\Tp{\xb_S} \Sigma^* \xb_S + \Tp{\xb_S} \Sigma^* \xb_{\Sc}} \notag \\
	& \geq \Abs{\Tp{\xb_S} \Sigma^* \xb_S} - \Abs{\Tp{\xb_S} \Sigma^* \xb_{\Sc}} \notag  \\
	& \geq \eigmax(\Sigma^*) \NormII{x_S}^2 - \Abs{\sum_{i \in S} \sum_{j \in \Sc} \Sigma^*_{i,j} x_i x_j} \notag  \\
	& \geq \eigmax(\Sigma^*) \NormII{x_S}^2 - \komax \sum_{i \in S} \sum_{j \in \Sc} \Abs{x_i x_j} \notag \\
	&= \eigmax(\Sigma^*) \NormII{x_S}^2 - \komax \NormI{x_S} \NormI{x_{\Sc}} \notag  \\
	&\Ineq[(c)]{\geq} \eigmax(\Sigma^*) \NormII{x_S}^2 - \komax \NormI{x_S}^2, \label{eq:fsp_temp1}
\end{align}
where (c) follows from the fact that $\NormI{x_{\Sc}} = \NormI{\betah_{\Sc} - \beta^*_{\Sc}} \leq \NormI{x_S}$ \eqref{eq:fsp_l1_ineq1}.
We also have the following upper bound:
\begin{align}
\Abs{\Tp{\xb_S} \Sigma^* x} \leq \NormI{\xb_S} \NormInfty{\Sigma^* x} \leq \sqrt{\Abs{S}} \NormII{x_S} \NormInfty{\Sigma^* x}
	\label{eq:fsp_temp2}.
\end{align}
From \ref{eq:fsp_temp1} and \eqref{eq:fsp_temp2} and under condition $\komax \leq \frac{\eigmin(\Sm{1}) \eigmin(\Sm{2})}{2 \Norm{\DO}_0}$
 we have that:
\begin{align*}
\NormI{x_S} = \NormI{\betah_S - \beta^*_S} \leq \frac{\NormInfty{\Sigma^* (\betah - \beta^*)}}{\komax}.
\end{align*}
Combining the above with \eqref{eq:fsp_l1_ineq2} and \eqref{eq:fsp_main_ineq} we get the following upper bound:
\begin{align}
\NormI{\betah - \beta^*} \leq \frac{3\NormInfty{\Sigma^* (\betah - \beta^*)}}{\kdmin}. \label{eq:fsp_l1_bound}
\end{align}

Next, we will upper bound $\NormInfty{\Sigma^* (\betah - \beta^*)}$.
\begin{align*}
\NormInfty{\Sigma^* (\betah - \beta^*)} 
	&= \NormInfty{ \Sigma^* \betah - b^* - (\Sigma^* \beta^* - b^*)} \\
	&\Ineq[(d)]{=} \NormInfty{ (\Sigma^* - \Smh{}) \betah + \Smh{} \betah - \bh + \bh - b^*} \\
	&\Ineq[(e)]{\leq} \lambda_n + \NormInfty{\bh - b^*} + \NormInfty{(\Sigma^* - \Smh{}) \betah} \\
	&\Ineq[(f)]{\leq} \lambda_n + \NormInfty{\bh - b^*} + \NormMax{\Sigma^* - \Smh{}} \NormI{\beta^*} \\
	&\Ineq[(g)]{\leq} 2 \lambda_n,
\end{align*}
where in (d) we used the fact that $\Sigma^* \beta^* - b^* = 0$, (e) follows from the fact that $\betah$
is the solution to the optimization problem \eqref{eq:opt} and triangle inequality, (f) follows from
the feasibility of $\beta^*$ for the optimization problem \eqref{eq:opt} and Cauchy-Schwartz inequality, and (g) follows
from the assumption that $\lambda_n \geq \NormInfty{\bh - b^*} + \NormMax{\Sigma^* - \Smh{}} \NormI{\beta^*}$.
Thus from \eqref{eq:fsp_l1_bound} and above, we get the following bound on the estimation error:
\begin{align}
\NormInfty{\beta^* - \betah} \leq \frac{6 \lambda_n}{\kdmin} \label{eq:fsp_final_bound}. 
\end{align}

Next, we will use concentration of Gaussian covariance matrix results from \cite{ravikumar2011high} to bound
$\NormInfty{\bh - b^*}$ and $\NormMax{\Sigma^* - \Smh{}}$. Note that
\begin{align*}
\NormInfty{\bh - b^*} = \NormMax{ \Smh{(1)} - \Smh{(2)} + \Sm{1} - \Sm{2}} \leq 
	2 \max_{\kappa \in \Set{1,2}} \NormMax{ \Smh{(\kappa)} - \Sm{\kappa}},
\end{align*}
where $\Sm{1}$ and $\Sm{2}$ are the true covariance matrices corresponding to the true SEMs.
From Lemma 1 of \cite{ravikumar2011high} and 
a union bound over $p^2$ entries of each $\Smh{1}$ and $\Smh{2}$, we have that with probability at least $1 - \delta$,
for some $\delta \in (0, 1)$:
\begin{align}
\NormInfty{\bh - b^*} \leq 2 \sqrt{\frac{c}{n} \log \frac{4p^2}{\delta}},
\end{align}
where $c = 3200 \max_{\kappa \in \Set{1,2}} \max_{i \in [p]} (\Sm{\kappa}_{i,i})^2$.
Next, we will bound $\NormMax{\Sigma^* - \Smh{}}$. 
\begin{align*}
\NormMax{\Sigma^* - \Smh{}} &= \max_{(a,b,c,d) \in [p]^4} \Abs{(\Sm{1}_{a,b}) (\Sm{2}_{c,d}) - (\Smh{(1)}_{a,b}) (\Smh{(2)}_{c,d})} \\
	&= \max_{(a,b,c,d) \in [p]^4}  \Abs{ \Sm{2}_{c,d}(\Sm{1}_{a,b} - \Smh{1}_{a,b}) +  
		\Smh{1}_{a,b} (\Sm{2}_{c,d} - \Smh{2}_{c,d}) }  \\
	&\leq \Abs{\Sm{2}_{c,d}} \Abs{\Sm{1}_{a,b} - \Smh{(1)}_{a,b}} +
		\Abs{\Sm{1}_{a,b}} \Abs{\Sm{2}_{c,d} - \Smh{(2)}_{c,d}} +
		\Abs{\Sm{1}_{a,b} - \Smh{(1)}_{a,b}} \Abs{\Sm{2}_{c,d} - \Smh{(2)}_{c,d}}
\end{align*}
From Lemma 1 of \cite{ravikumar2011high} we have that with probability at least $1 - 2 \delta$ the
following hold:
\begin{align*}
\Abs{\Sm{1}_{a,b} - \Smh{(1)}_{a,b}} &\leq \sqrt{\frac{c_1}{n} \log \frac{4}{\delta}} \\
\Abs{\Sm{2}_{c,d} - \Smh{(2)}_{c,d}} &\leq \sqrt{\frac{c_2}{n} \log \frac{4}{\delta}} ,
\end{align*}
where $c_{\kappa} = 3200 \max_{i \in [p]} (\Sm{\kappa}_{i,i})^2$. Since, $n \geq \log (\nicefrac{4}{\delta})$
and taking a union bound over $2p^2$ entries of the empirical covariance matrices we get that with 
probability $1 - \delta$, for some $\delta \in (0,1)$:
\begin{align*}
\NormMax{\Smh{} - \Sigma^*} \leq c' \sqrt{\frac{c}{n} \log \frac{8p^2}{\delta}},
\end{align*}
where $c' = \sqrt{c} + \NormMax{\Sm{1}} + \NormMax{\Sm{2}}$. This implies the lower bound on
$\lambda_n$ is given as follows:
\begin{align*}
\lambda_n \geq 2 \sqrt{\frac{2c}{n} \log \frac{2p^2}{\delta}} + c' \NormI{\beta^*} \sqrt{\frac{3c}{n} \log \frac{2p^2}{\delta}}
	\geq C \sqrt{\frac{1}{n} \log \frac{2p}{\delta}},
\end{align*}
where the constant $C$ is given in the statement of the theorem.
Setting the estimation error to be at most $\varepsilon$ implies the following upper bound on $\lambda_n$:
\begin{align*}
\lambda_n \leq \frac{\kdmin \varepsilon}{6}.
\end{align*}
Setting $n$ as given in the theorem ensures that the lower bound in less than the upper bound.

Lastly, we show that the $\beta^*$ is feasible. We have the following:
\begin{align*}
\NormInfty{\Smh{} \beta^* - \bh} 
	&=  \NormInfty{(\Smh{} - \Sigma^*) \beta^* - (\bh - b^*) + \Sigma^* \beta^* - b^*} \\
	&\leq \NormMax{\Smh{} - \Sigma^*} \NormI{\beta^*} + \NormInfty{\bh - b^*} \\
	&\leq \lambda_n,
\end{align*}
where the second line follows from the fact that $\Sigma^* \beta^* - b^* = 0$ and the triangle inequality,
and the last line follows from the assumption on $\lambda_n$.
\end{proof}
\begin{proof}[Proof of Corollary \ref{cor:finite_sample}]
To prove the corollary, we just need to show that for any subset $S \subseteq [p]$ estimating
the difference of precision matrix $\D_{\Omega}^S$ using \eqref{eq:opt}, satisfies Assumption \ref{ass:oracle}.
For any subset $S \subseteq [p]$ denote the corresponding covariance matrices by $\Sm{\kappa,S}$, for $\kappa \in \Set{1,2}$.
Let
$\komaxS \defeq \max \Set{ \Abs{ \Sm{1,S}_{i,j} \Sm{2,S}_{k,l}} \mid i,j,k,l \in S, \, (i,j) \neq (k,l)}$
and $\kdminS \defeq \min \Set{ \Sm{1,S}_{i,i} \Sm{2,S}_{i,i} \mid i \in [S]}$. Since, for any $S$
$\eigmin(\Sm{\kappa,S}) \geq \eigmin(\Sm{\kappa})$ for $\kappa \in \Set{1,2}$, and $\komaxS \leq \komax$, we have:
\begin{align*}
\komaxS \leq \komax \Ineq[(a)]{\leq} \frac{\eigmin(\Sm{1}) \eigmin(\Sm{2})}{2d} \leq 
	\frac{\eigmin(\Sm{1,S}) \eigmin(\Sm{2,S})}{2\Norm{\D_{\Omega}^S}_0},
\end{align*}
where $(a)$ follows from the Assumption in the corollary. Therefore, we have that estimating
$\D_{\Omega}^S$ using \eqref{eq:opt} satisfies the incoherence condition (Theorem \ref{thm:cai_precision_difference}) for each $S$. Next,
since the constant $C = \BigO{d}$ for any set $S$, we have that the condition on the number
of samples and regularization parameter (Theorem \ref{thm:cai_precision_difference}) are satisfied as well.
Next from Lemma 1 \cite{ravikumar2011high} we have 
that with probability at least $1 - \delta$ we have that simultaneously for both $\kappa \in \Set{1,2}$:
\begin{align*}
\NormMax{\Sm{\kappa} - \Smh{(\kappa)}} \leq \sqrt{\frac{c}{n} \log \frac{4p^2}{\delta}},
\end{align*}
where $c = 3200 \max_{\kappa \in \Set{1,2}} \max_{i \in [p]} (\Sm{\kappa}_{i,i})^2$. Thus, we have that
with probability at least $1 - \delta$ and simultaneously for all $S \subseteq [p]$ and $\kappa \in \Set{1,2}$
\begin{align*}
\NormMax{\Sm{\kappa,S} - \Smh{(\kappa,S)}} \leq \sqrt{\frac{c}{n} \log \frac{4p^2}{\delta}}.
\end{align*}
From the proof of Theorem \ref{thm:cai_precision_difference} it is clear that using $\Smh{(\kappa)}$ in \eqref{eq:opt}
such that $\NormMax{\Sm{\kappa} - \Smh{(\kappa)}} \leq \sqrt{\frac{c}{n} \log \frac{4p^2}{\delta}}$, we get an
estimate $\widehat{\D}_{\Omega}^S$, by solving \eqref{eq:opt}, which satisfies 
$\NormMax{\widehat{\D}_{\Omega}^S - \D_{\Omega}^S} \leq \varepsilon$. Combined with the fact that for each $S$
the condition required as per Theorem \ref{thm:cai_precision_difference} for $n$, $\lambda_n$, and $\komaxS$
are satisfied, we get that the finite sample algorithm that uses \eqref{eq:opt} to estimate the difference of
precision (sub-)matrices satisfies Assumption \ref{ass:oracle}. Thus, the final claim follows as per Theorem \ref{thm:finite_sample}.
\end{proof}
\begin{proof}[Proof of Theorem \ref{thm:lower_bound}]
Given two \emph{graphs} $G = ([p], E)$ and $G' = ([p], E')$, let $G \oplus G' \defeq ([p], (E \setminus E') \union (E' \setminus E))$
denote the graph obtained by taking edges exclusive to the graphs $G$ and $G'$ and removing edges common to them.
Let $\Gf$ be the set of DAGs over $[p]$ variables and for any DAG $G \in \Gf$ let $\DGf(G)$ be the set of directed ``difference'' graphs
having the same causal order as $G$. Let $\DGf \defeq \Set{\DGf(G) \mid G \in \Gf}$. Given a DAG $G$, let $P_G$ denote the
distribution induced by the Gaussian linear SEM $(B(G), D)$ where the edge weight matrix $B(G)$ is given as follows:
\begin{align*}
B(G)_{i,j} = \left\{
\begin{array}{ll}
\nicefrac{1}{\sqrt{\Abs{\nPar(i)}}} & j \in \nPar(j), \\
0 & \text{otherwise}
\end{array}
\right.
\end{align*}
and $D = \var I_p$, where $I_p$ is the $p \times p$ identity matrix. 
Thus, given a DAG $G$ the distribution over the variables $\Set{X_1, \ldots, X_p}$
is uniquely defined. Let $\Pf(\Gf, \DGf) = \Set{P_{\G{1}} \times P_{\G{1} \oplus \D_G} \mid \G{1} \in \Gf, \D_G \in \DGf(\G{1})}$ 
be the set of distributions  corresponding to the graph families $\Gf$ and $\DGf$. Since $\Pf(\Gf, \DGf) \subset \Pf$ and
finite, the minimax error is lower bounded as follows:
\begin{align*}
\perr \geq \inf_{\Dec} \max_{P \in \Pf(\Gf, \DGf)} \Prob[\Xf \sim P^n]{\Dec(X) \neq \D_G(P)}.
\end{align*}

Let $\Gfr \subseteq \Gf$ and
$\forall G \in \Gfr$ $\DGfr(G) \subseteq \DGf(G)$. Also, let $\DGfr = \Set{\DGfr(G) \mid G \in \Gfr}$.
Since $\Pf(\Gfr, \DGfr) \subseteq \Pf(\Gf, \DGf)$, the minimax error is further lower bounded
as follows:
\begin{align*}
\perr \geq \inf_{\Dec} \max_{P \in \Pf(\Gfr, \DGfr)} \Prob[\Xf \sim P^n]{\Dec(\Xf) \neq \D_G(P)}.
\end{align*}
We will construct the restricted ensembles $\Gfr$ and $\DGfr$ as follows. Each $G = ([p], E) \in \Gfr$ is a 
fully connected directed bipartite graph. That is, we partition the set $[p]$ in to $U$ and $V$
such that $[p] = U \union V$, $U \intersection V = \emset$, $\Abs{U} = \ceil{\nicefrac{p}{2}}$ and $\Abs{V} = \floor{\nicefrac{p}{2}}$.
The edge set $E = \Set{(v, u) \mid u \in U \text{ and } v \in V }$. Note that $(v, u)$ denotes the directed edge $v \leftarrow u$.
We will denote the graph $G$ by $(V, U, V \times U)$.
For any $G = (V, U, V \times U) \in \Gfr$, a DAG $G' = (V, U, E) \in \DGfr$ is generated as follows. For each node $v \in V$
we randomly pick a subset $U(v) \subset U$ such that $\Abs{U(v)} = \dt$ and set the parents of $v$ to be $U(v)$. Therefore, the edge
set $E' = \Set{(v, u) \mid v \in V, u \in U(v), U(v) \subset U, \Abs{U(v)} = \dt}$. Thus, the graph $G'$ is $\dt$-sparse. Note that the
graph $G \oplus G'$ is a fully connected bipartite DAG from which the edges in $G'$ have been deleted. 

Note that for any $G \in \Gfr$, $\Abs{\DGfr(G)} = {\ceil{\frac{p}{2}} \choose \dt}^{\floor{\frac{p}{2}}}$.

We will be using the following results from \cite{ghoshal2017information}.

\begin{theorem}[Generalized Fano's inequality Theorem 1 of \cite{ghoshal2017information}]
\label{thm:fano}
Let $W, X,$ and $Y$ be random variables such that the conditional independence relationship
between them are given by the following partially directed graph:
\begin{center}
\includegraphics[width=0.2\linewidth]{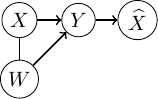}
\end{center}
Let $\what{X}$ be any estimator of $X$. Then,
\begin{align}
        \Prob{X \neq \what{X}} \geq 1 - \frac{\MI(Y ; X | W) + \log 2}{H(X|W)}. \label{eq:thm_fano}
\end{align}
\end{theorem}
Nature generates a DAG $G$ uniformly at random from the family $\Gfr$ and then generates a difference DAG $\D_G$
uniformly at random from the family $\DGfr(G)$. The minimax error can further be lower bounded as follows:
\begin{align}
\perr &\geq \inf_{\Dec} \max_{P \in \Pf(\Gfr, \DGfr)} \Prob[\Xf \sim P^n]{\Dec(\Xf) \neq \D_G(P)} \notag \\
	&\geq \inf_{\Dec} \Exp[P]{\Prob{\Dec(\Xf) \neq \D_G(P)}} \notag \\
	&= \inf_{\Dec} \Exp[G]{\Exp[\D_G]{\Prob{\Dec(\Xf) \neq \D_G}} \mid G}, \label{eq:perr_lb} \\	
\end{align}
where in the second line the probability is over both $\Xf \sim P^n$ and $P$ drawn from the family $\Pf(\Gfr, \DGfr)$.
Next, we lower bound $\Prob[\Xf \sim P_{G}^n \times P_{G \oplus \D_G}^n]{\Dec(\Xf) \neq \D_G}$ using Theorem \ref{thm:fano}.
Towards that end we need to first upper bound the mutual information $\MI(\Xf ; \D_G \mid G)$ and compute the conditional entropy
$H(\D_G \mid G)$. Let $Q = \mathcal{N}(0, I)$ be the standard isotropic Gaussian distribution over $X$. We upper bound $\MI(\Xf ; \D_G \mid G)$
by adapting Lemma 4 from \cite{ghoshal2017information} for our purpose which we state below.
\begin{lemma}[Lemma 4 of \cite{ghoshal2017information}]
Let $P_{G, \D_G} = P_G \times P_{G \oplus \D_G}$ denote the data distribution 
given a specific DAG $G$ and a specific difference DAG $D_G$. Then,
\begin{align*}
\MI(\Xf; \D_G \mid G) \leq \frac{1}{\Abs{\DGfr(G)}} \sum_{\D_G \in \DGfr(G)} \KL{P_{G, \D_G}^n}{(Q')^n},
\end{align*}
where $Q' = \mathcal{N}(0, I)$ is the standard $2p$-dimensional isotropic Gaussian distribution.
\end{lemma}
From the above we have that
\begin{align*}
\MI(\Xf; \D_G \mid G) &\leq \frac{1}{\Abs{\DGfr(G)}} \sum_{\D_G \in \DGfr(G)} \KL{P_{G, \D_G}^n}{(Q')^n}  \\
	&=  \frac{n}{\Abs{\DGfr(G)}} \sum_{\D_G \in \DGfr(G)} \KL{P_{G}}{Q} + \KL{P_{G \oplus \D_G}}{Q}
\end{align*}
Note that the distribution indexed by a DAG $G$ is $P_{G} = \mathcal{N}(0, \Sigma(G))$, with $\Sigma(G) = (I - B(G))^{-1} D \Tpn{(I - B(G))}$.
$Q$ is the $p$-dimensional standard isotropic Gaussian distribution. From the KL-divergence
characterization for multivariate Gaussian distribution we have that:
\begin{align*}
\KL{P_{G}}{Q} = \frac{1}{2} \left \{ \trace(\Sigma(G)) -p - \ln \Abs{\Sigma(G)} \right\}
\end{align*}
For any $u \in U$, $\Var{X_u} = \var$, while for a $v \in V$, $\Var{X_v} = (\nicefrac{1}{\sqrt{\Abs{U}}}) \sum_{u \in U} X_u + \var = 2 \var$
since $X_u$'s are independent for all $u \in U$. Therefore, 
$\trace(\Sigma(G)) = \ceil{\nicefrac{p}{2}} \var + 2 \floor{\nicefrac{p}{2}} \var \leq (\nicefrac{3}{2}) p\var$.
Next,
\begin{align*}
\Abs{\Sigma(G)} &= \Abs{(I - B(G))^{-1} D \Tpn{(I - B(G))}}  \\
	&= \frac{1}{\Abs{(I - B(G)) D^{-1} \Tp{(I - B(G))} }} = \frac{1}{\Abs{D^{-1}}} = \Abs{D} = (\var)^p
\end{align*}
Setting $\var = \nicefrac{2}{3}$, we have that 
$\KL{P_{G}}{Q} = \nicefrac{1}{2} \{ p - p + p \ln \nicefrac{3}{2} \} < (\nicefrac{p}{2}) \ln \nicefrac{3}{2}$.
Since each vertex in the DAG $\D_G$ has exactly $p - \dt$ parents, $\KL{P_{G \oplus \D_G}}{Q} \leq (\nicefrac{p}{2}) \log 2$. 
Therefore,  $\MI(\Xf; \D_G \mid G) \leq np \ln \nicefrac{3}{2} < \nicefrac{np}{2}$.

Finally, since $\D_G$ is picked uniformly at random from the set $\DGfr(G)$ and $G$ itself is also picked uniformly at random from $\DGf$, we have that $H(\D_G \mid G) = \floor{\nicefrac{p}{2}} \log {\ceil{\nicefrac{p}{2}} \choose \dt}
\leq (\nicefrac{p\dt}{2}) \log (\nicefrac{p}{2\dt})$.
Therefore, from \eqref{eq:perr_lb} and Theorem \ref{thm:fano} we have that
\begin{align*}
\perr \geq 1 - \frac{\nicefrac{np}{2} + \log 2}{(\nicefrac{p\dt}{2}) \log (\nicefrac{p}{2\dt})}.
\end{align*}
From the above we have that $\perr \geq \frac{1}{2}$ if $n \leq \frac{\dt}{2} \log \frac{p}{2\dt} - \frac{2}{p} \log 2$.
\end{proof}

\section{Comparison with indirect methods}
\label{app:comparison}
\citet{ghoshal2018learning} have obtained optimal sample complexity guarantees for estimating linear SEMs with equal noise variance. Therefore, we compare our method against using their
method to first estimate the individual DAGs and then compute
their differences. Using the method of \citet{ghoshal2018learning}
would require at least 
\begin{align*}
    \min(n_1, n_2) \geq 2 (\nicefrac{C_1 C}{\varepsilon})^2 
        \log (\nicefrac{2p}{\sqrt{\delta}}),
\end{align*}
samples to estimate the individual DAGs with probability at least $2\delta$, where the constant $C_1 = \BigO{\max_{\kappa, i}
\Sm{\kappa}_{i,i}}$ and 
$C = \BigO{\max_{\kappa}\NormI{\Om{\kappa}}^2}$. Therefore, the
sample complexity of indirectly estimating the difference DAG 
depends on $(\max_{\kappa, i} \Sm{\kappa}_{i,i} \NormI{\Om{\kappa}}^2 )^2$. In comparison, the sample complexity of our direct estimation method depends on $(\Abs{\DO}_1 \max_{\kappa, i} \NormMax{\Sm{\kappa}} \Sm{\kappa}_{i,i})^2$.
Note that our sample complexity also depends on $1/\kdmin$,
where $\kdmin = \min_{i} \Sm{1}_{i,i}\Sm{2}_{i,i}$. 
Since $\Sm{1}_{i,i} \geq \sigma_i^2$, $\kdmin$ can be ignored as a
constant. While our method requires a mutual incoherence condition, their method for estimating the individual DAGs require the noise
variances to be the same (or close to being the same). Therefore,
if the incoherence condition is satisfied, our method has strictly
lower sample complexity which depends only on the sparsity of
the difference between the moral graphs of the two SEM.

\section{Additional Experiments}

\subsection{Logarithmic dependence on the number of variables}
\label{app:log_phase}
In this section, we show the $\BigO{\log p}$ dependence of the sample complexity, as prescribed by Corollary \ref{cor:finite_sample} and Theorem \ref{thm:lower_bound}.
In Figure \ref{fig:log_phase}, we perform 30 repetitions to estimate the probability of exact recovery, where the SEM pairs are sampled according to Section \ref{sec:exps}.
Then, for each pair, we obtain $\floor{e^C \log p}$ samples for $C \in [3,12]$ and $p \in \{5,10,15\}$.
\begin{figure}[!ht]
\centering
\includegraphics[width=0.45\linewidth]{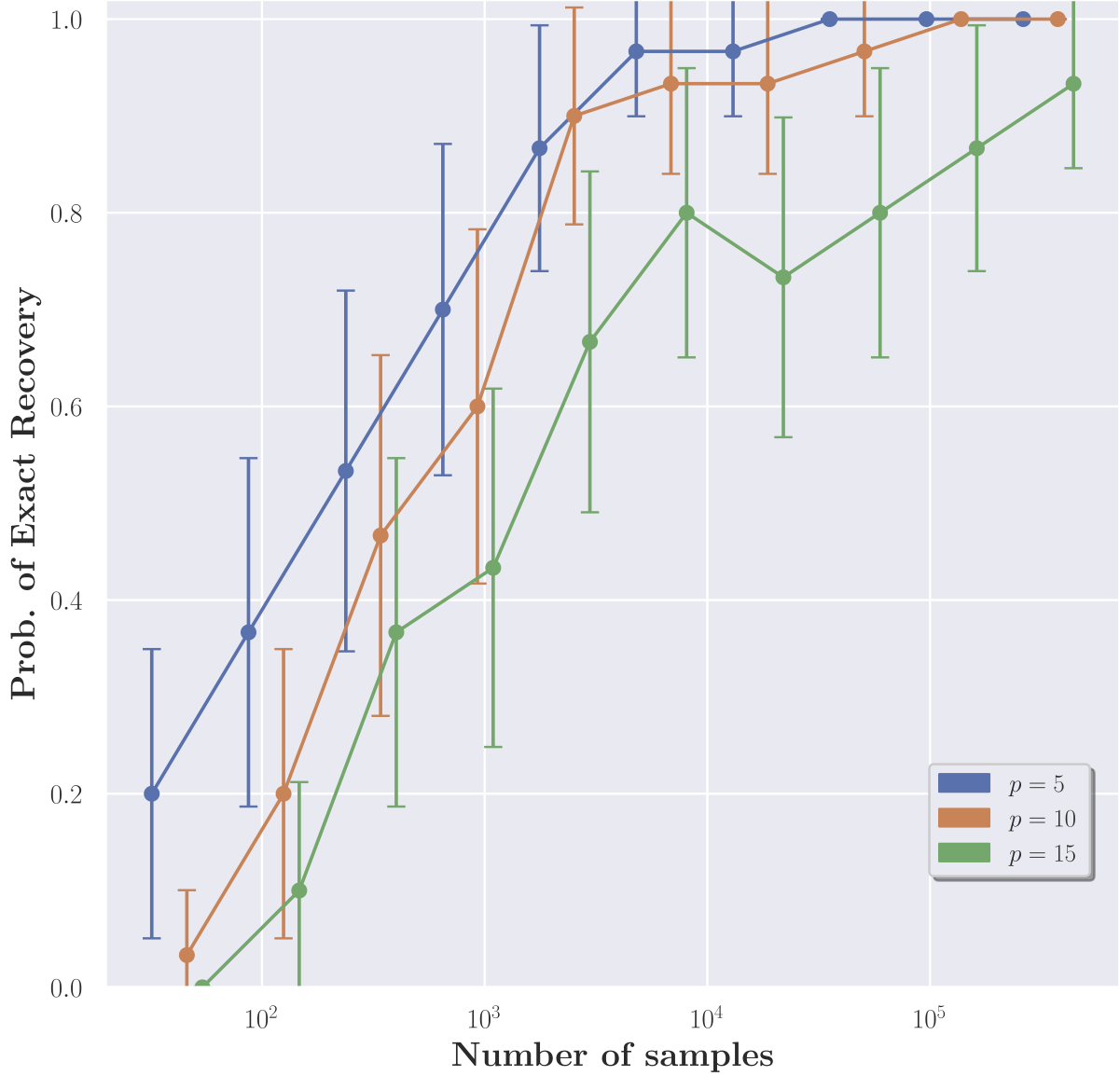} 
\caption{Probability of exact recovery computed across 30 repetitions.
The number of samples is set to $\floor{e^C \log p}$ for $C \in [3,12]$ and $p \in \{5,10,15\}$.
The x-axis is in log-scale.
We can observe the logarithmic dependence on $p$, as prescribed by Corollary \ref{cor:finite_sample} and Theorem \ref{thm:lower_bound}.
}
\label{fig:log_phase}
\end{figure}

\subsection{Real-world experiment}
\label{app:real_world_exp}
In this section we test our algorithm in the medical domain.
The 1000 functional connectomes dataset contains resting-state fMRI of 1128 subjects collected on 41 sites around the world. 
The dataset is publicly available at \url{http://www.nitrc.org/projects/fcon_1000/}. 
Resting-state fMRI is a procedure that captures brain function of a subject that is at wakeful rest (i.e., not focused on the outside world). 
Registration of the dataset to the same spatial reference template (Talairach space) and spatial smoothing was performed in SPM (\url{http://www.fil.ion.ucl.ac.uk/spm/}). 
We extracted voxels from the gray matter only, and grouped them into 157 regions (i.e., $p=157$) by using standard labels, given by the Talairach Daemon (\url{http://www.talairach.org/}). 
These regions span the entire brain: cerebellum, cerebrum and brainstem. 
In order to capture laterality effects, we have regions for the left and right side of the brain.

A relevant neuroscientific aim is to discover the changes in the default mode network (which is active during wakeful rest) in subjects who had the eyes open, versus subjects who had the eyes closed. 
This is known to make a significant difference in brain activity for activity-oriented tasks, but this is unclear in resting state. 
Our method recovered a difference DAG between brain regions that belong to the visual cortex (ventral and dorsal visual streams) in the back of the brain (see Figure \ref{fig:real_world_exp}). 
Thus, besides having strong theoretical guarantees, our method has the potential to also produce meaningful results in practice.
\begin{remark}
We emphasize that our goal here is to demonstrate the practicality of our method.
It is beyond the scope of this work to make scientific discoveries in neuroscience.
\end{remark}
\vspace{-0.08in}
\begin{figure}[!ht]
    \centering
    \includegraphics[width=0.75\linewidth]{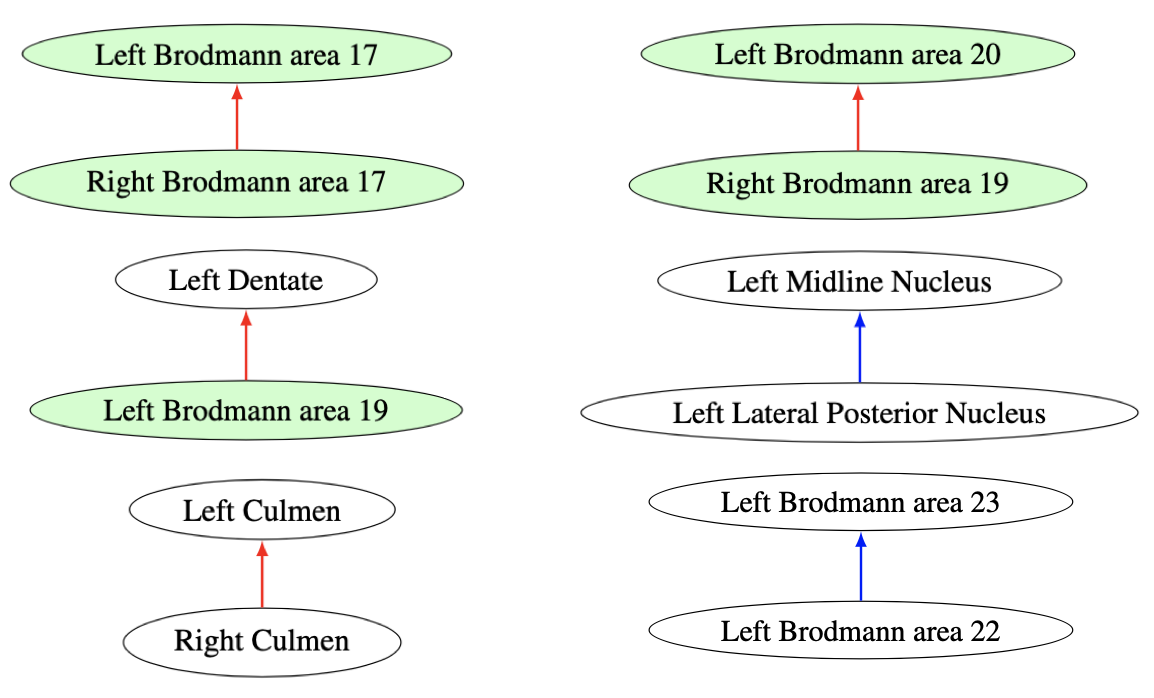}
    \caption{
    Out of the 157 variables, only 12 variables participate in the difference DAG, i.e., the remaining 145 variables are disconnected.
    In this case, $\epsilon$ was set to $0.1$ for our finite-sample version of Algorithm \ref{alg:exact_alg} (see Remark \ref{remark:epsilon}).
    The red edges indicate important connections discovered.
    The visual cortex (nodes colored in green) includes: Brodmann areas 17, 19, 20.
    We note that Dentate is partially involved in the visuospatial function.
    Lateral Posterior Nucleus is also involved in vision, while Culmen is believed to mediate in visual reflexes.
    }
    \label{fig:real_world_exp}
\end{figure}

\end{appendix}

\end{document}